\documentclass[11pt,a4paper]{article}

\usepackage[utf8]{inputenc}
\usepackage[margin = 1in]{geometry}
\usepackage{graphicx, xcolor}
\usepackage{amsthm, amsmath, amssymb, amsthm}
\newtheorem{theorem}{Theorem}[section]

\usepackage{hyperref}
\usepackage{lipsum}
\usepackage{comment}
\usepackage{caption, subcaption}
\usepackage{multirow}
\usepackage{longtable}
\usepackage{array}
\usepackage[ruled, vlined]{algorithm2e}
\usepackage{tikz}
\usepackage{pgfplots}
\pgfplotsset{compat = 1.14}
\usetikzlibrary{shapes.geometric, arrows}

\newcommand{\Var}{\mathbb{V}\text{ar}}

\newcommand{\indep}{\perp\!\!\!\perp}

\tikzstyle{obsvars} = [rectangle, rounded corners, minimum width=1cm, minimum height=1cm,text centered, draw=black, fill=green!30]
\tikzstyle{latents} = [rectangle, rounded corners, minimum width=1cm, minimum height=1cm,text centered, draw=black, fill=orange!30]
\tikzstyle{invisible} = [rectangle, rounded corners, minimum width=1cm, minimum height=1cm, text centered, draw=white, fill=white!30]
\tikzstyle{arrow} = [thick,->,>=stealth]

\newcolumntype{L}[1]{>{\raggedright\let\newline\\\arraybackslash\hspace{0pt}}m{#1}}
\newcolumntype{C}[1]{>{\centering\let\newline\\\arraybackslash\hspace{0pt}}m{#1}}
\newcolumntype{R}[1]{>{\raggedleft\let\newline\\\arraybackslash\hspace{0pt}}m{#1}}

\title{The Information Content of Taster's Valuation \\ in Tea Auctions of India}

\author{Abhinandan Dalal\footnote{Student, Indian Statistical Institute Kolkata} , Diganta Mukherjee\footnote{Sampling and Official Statistics Unit, Indian Statistical Institute, Kolkata} \, and Subhrajyoty Roy\footnote{Student, Indian Statistical Institute Kolkata}   }
\date{\today}

\begin{document}

\maketitle

\begin{abstract}
    
    Tea Auctions across India occur as an ascending open auction, conducted online. Before the auction, a sample of the tea lot  is sent to potential bidders and a group of tea tasters. The seller’s reserve price is a confidential function of the tea taster’s valuation, which also possibly acts as a signal to the bidders.

    In this paper, we work with the dataset from a single tea auction house, J Thomas, of tea dust category, on 49 weeks in the time span of 2018-2019, with the following objectives in mind:
    \begin{itemize} 
    \item Objective classification of the various categories of tea dust (25) into a more manageable, and robust classification of the tea dust, based on source and grades.
    \item Predict which tea lots would be sold in the auction market, and a model for the final price conditioned on sale.
    \item To study the distribution of price and ratio of the sold tea auction lots.
    \item Make a detailed analysis of the information obtained from the tea taster's valuation and its impact on the final auction price.  
    \end{itemize}
    
    The model used has shown various promising results on cross-validation. The importance of valuation is firmly established through analysis of causal relationship between the valuation and the actual price. 
    The authors hope that this study of the properties and the detailed analysis of the role played by the various factors, would be significant in the decision making process for the players of the auction game, pave the way to remove the manual interference in an attempt to automate the auction procedure, and improve tea quality in markets.

\end{abstract}

\begin{center}
    
\end{center}
\pagebreak

\tableofcontents
\pagebreak

\section{Introduction}

Tea (Camellia sinensis) is a manufactured drink that is consumed across the world. The  tea  crop  has  rather specific agro-climatic requirements that are only available in tropical and  subtropical climates. Tea production, therefore, is geographically limited to a few areas around the world and is highly sensitive to changes in  growing  conditions. Majority of the tea producing  countries are located in the continent of Asia with China, India,  Kenya,  Sri  Lanka  and  Vietnam being  the top  producers  (in  that  order),  accounting  for around  78\% of the  world tea  production and 73\% of exports.  World  tea  production  is  estimated at over 5 million tonnes in 2015, valued around Rs 1 trillion. It increased by 4 percent to 5.2 million  tonnes  in  2014 (see \cite{FAO1} for further details).
China remains  the  largest  tea producing  country  with  an output of 2.1 million  tonnes in 2014, accounting for more than 40 percent of the world total, while production in India, the second largest producer, remained flat at 1.2 million tonnes in 2014, contributing  around  30\%  of  would  production.  

The  tea  industry  is  one  of  the  oldest  organized industries in India with a large network of tea producers, retailers, distributors, auctioneers, exporters and packers. Interestingly, India is also the world's largest consumer of black tea with the domestic market consuming around 1,000 million kg of tea during 2016. India’s annual production of tea is around 1,200 million kgs and the market size is estimated to be approximately Rs 20,000 Crore.
The  Tea  Industry  in  India  also  derives  its  importance  by  being  one of  the major  foreign  exchange earners and for playing a vital role towards employment generation as the industry is highly labour intensive. India exports around 225 million kg of tea and is the fourth largest exporter in the world with Russia being its largest importer. The annual value of tea exports from India is around \$800 million. The other major importers of Indian tea are Iran, UK, Pakistan and UAE. 

Tea is heterogeneous both over season and region, and even intra-region. Varieties of tea: 90\%  of  the tea  produced  in  India  is  of  CTC  (Cut, Tear and Curl) variety followed by Orthodox (9\%) and  Green  Tea  (1\%).  CTC  tea  is  largely  graded  as Broken  Leaf,  Dust  and  Fannings.  Each  of  these grades has a dozen of sub-grades based on the size of the grain etc.
Several factors influence the demand for tea, including the price and income variables, demographics such as age, education, occupation, and cultural background. Apart from consumption, other main drivers of international tea prices are trends and changes in per capita consumption, market access, the potential effects of pests and diseases on production, and changing dynamics between retailers, wholesalers and multinationals (Source: \cite{FAO1}). 

Tea demand is very price sensitive. Price elasticities for black tea vary between -0.32 and -0.80, which means that a 10 percent increase in black tea retail prices will lead to a decline in demand for black tea  between  3.2  percent and  8  percent,  according  to  FAO.  The average  weekly  volatility  of  CTC  grade prices in Siliguri is 7\%. Though the prices seem to be volatile, the tea prices show a specific pattern. The prices are at a peak as the new crop arrivals begin, with an increase in supplies, the prices witness a decline, lasting till the end of the season.

Due to the variation of quality, tea within the same grades are sold at a very wide range of prices in auctions. For example on any given auction day, say Broken Pekoe (BP) variety would sell at a range starting from as low as Rs 60  to Rs 250 per kg, depending on the producer mark (brand), quality and demand, making standardization difficult. Moreover, a given quality of tea is not available throughout the year. \footnote{The details regarding the tea industry has been collected from \cite{FAO1}, \cite{FAO2} and website \cite{FAO3}.}

{\bf Overview of e-auctions of tea:}
An e-auction is a primary marketing channel for selling tea to the highest bidder. The auction system serves two basic purposes. The first purpose is to facilitate price discovery by bringing the buyers and sellers  to  a  common  platform  with  broker’s intermediation. Buyers bid for lots of tea and each lot is sold to the winning bidder. The second purpose is that  the  auction  system  provides  a  guaranteed transaction  protocol  for  the  transaction.  The transaction includes activities such as delivery of tea to the warehouse, storing, sampling, bidding and payment (see \cite{FAO4} for a discussion). 

Since September 2016, the auctions are pan India. It means that a member registered with any tea trade association anywhere in India can directly participate in any e-auctions conducted by Tea Board. Earlier, the buyer registered with local tea association could only buy from the e-auctions taking place in respective centers. The  tea  auction  system  brings  the  buyers  and  sellers  together,  to  determine  the  price  through interactive competitive bidding on the basis of prior assessment of quality of tea. Manufactured tea is dispatched  from  various  gardens/  estates  to  the  auction  centres  for  sale  through  the  appointed auctioneers,  on  receipt  of  which,  the  warehouse  keeper  sends  an  arrival  a ‘weighment  report’ showing  the  date of  arrival  and  other  details  pertaining  to the tea  including  any  damage or short receipt from the carriers. 
The tea is catalogued on the basis of their arrival dates within the framework of the respective Tea Trade  Associations,  the  quantities  are  determined  according  to  the  rate  of  arrivals  at  a  particular auction  centres.  Registered  buyers,  representing  both  the  domestic  trade  and  exporters  receive samples  of  each  lot  of  teas  catalogued,  which  is  generally  distributed  a  week  ahead  of  each  sale enabling the buyers to taste, inform their principals and receive their orders well in time for sale. 
The auctioneers taste and value the tea for sale and these valuations are released to the traders. 
Guidelines for the price levels likely to be established when the tea is sold are formulated on the basis of these valuations and last sale price. J. Thomas \& Co. Pvt. Ltd is the largest auctioneer in the world, handling over 200 million kg of tea a year, which is one-third of all tea auctioned in India. 

Given the above complexities, the report aims to evaluate the feasibility of automating the pricing process to the extent of dispensing with the manual testing and valuation steps. The next section describes the data set we have used for our analysis. We first discuss the clustering exercise according to Grade and source in Section 3.  Section 4 discusses the pattern of volumes over different months of the year. The pattern of salability of different tea lots is investigated in section 5. Section 6 discusses the Price to value ratio to gain some insight into the pricing pattern which finally is used fully in the pricing models developed in Section 7. A detailed analysis of the value - price causality question is made in Section 8. The final comments on the feasibility of automation and future plans are discussed in section 9.

\section{Description of data}
In this report, we have used J-Thomas datasets on the weekly tea  details for Kolkata Dust tea, Orthodox tea details, CTC tea details and Darjeeling tea details. Moreover, we have the e-auction statistics as a part and parcel of the dataset. We have used the data in 2018 for modelling, as the training data, and the 2019 data has been used for cross-validation. We begin our initial modelling assuming that the model, conditioned on the relevant factors, does not depend on the year of the auction. We shall later see in the cross-validation procedure, that our predictions are quite satisfactory to assert that our assumption was not falsified.

The e-auction statistics (2018-19) consists of the name of tea leaf type, total lots offered in auctions, total amount sold in packets and in quantity, and average price. For each such tea leaf type, detailed info on the weekly sale, total amount sold in packets and in quantity, and average price has also been provided. An initial overview of the characteristics of the data at hand, has been produced in \cite{our_earlier}.

In the J-Thomas datasets, we find lot numbers (hence the difference  between the maximum and minimum lot number would give us the number of lots offered), the categorical variable: the grade of the tea,  number of packages offered, the valuation given by the agency, and finally the auction selling price. 

In our dataset, we have 25 types of tea grades available \cite{names}, namely, 
\begin{center} 
    \begin{tabular}{|c|c||c|c|}
        \hline
        & & & \\
        CD & Churamani Dust 
        & CD1 & Churamani Dust 1
        \\ CHD & $\cdots$ 
        & CHD1 & $\cdots$ 
        \\ CHU & $\cdots$ 
        & D & Dust 
        \\  D(F) & Dust Fine 
        & D(SPL) & Dust Special
        \\ D1 & Dust 1
        & D1(SPL) & Dust 1 Special
        \\ GTDUST & Golden Tea Dust
        & OCD & Orthodox Churamani Dust
        \\  OD & Orthodox Dust 
        & OD(S) & Orthodox Dust (S) 
        \\ OD1 & Orthodox Dust 1
        & OPD & Orthodox Pekoe Dust
        \\ OPD(Clonal) & Orthodox Pekoe Dust (Clonal)
        & OPD1 & Orthodox Pekoe Dust 1
        \\ ORD & Orthodox Red Dust 
        & PD & Pekoe Dust 
        \\ PD(FINE) & Pekoe Dust (Fine)
        & PD(SPL) & Pekoe Dust Special
        \\ PD1 & Pekoe Dust 1
        & PD1(SPL) & Pekoe Dust Special 
        \\ RD1 & Red Dust 1 & & \\
        & & & \\
        \hline
    \end{tabular}
\end{center}

However, there are only three main broad categories for tea dust: \textbf{D1, PD} and \textbf{PD1} according to Wikipedia \cite{wiki}. Hence our grades require clustering.\par
Out of the 49 weeks of data we possess, 38 weeks of auctions are from 2018, from the period of January to December. However, since not every week had an auction, we thus have data on weeks 2-9, 18, 20, 22-38, 40-41,43-52. Along with this, the remaining data is on the first 12 weeks in 2019, (except the 11th week, in which no auction took place), which we shall be using for cross-validation. Hence, based on our training set, we club together the tea grades for which there are less than 38 observations in the entire dataset (as we require at least one observation per week on the average.) This clubbing is done by its proximity to the other tea grades, where the proximity is based on the qualitative similarity of the grades. Thus we obtain the following 14 clubbed grades till now.

\begin{itemize}
    \item \textbf{CD}: Consisting of CD, CHD and CHU
    \item \textbf{CD1}: CD1 and CHD1
    \item \textbf{D}: D and D special
    \item \textbf{D(FINE)}
    \item \textbf{D1}: D1 amd D1 Special
    \item \textbf{OCD}
    \item \textbf{OD}: OD and OD-special
    \item \textbf{OD1}
    \item \textbf{OPD}:OPD, OPD-Clonal and ORD
    \item \textbf{OPD1}
    \item \textbf{PD}: PD and PD-Special
    \item \textbf{PD(FINE)}
    \item \textbf{PD1}: PD1 and PD1-Special
    \item \textbf{RD1}
    
\end{itemize}
Due to the only packet of GT Dust in the dataset, and that too remaining unsold, and further, due to its lack of immediate similarity from any of the existing tea grades, we conclude that GT Dust is a very rare category in this dataset, and hence we leave it out for the classification problem for now.

\section{Classification and Clustering}

\subsection{Clustering based on Grade}
The number of clusters so formed in the previous subsection is still quite large, and we suspect that they have quite an inherent similarity in their characteristics and hence in their market appeal and corresponding auction transactions. To have an idea about this, we form a dissimilarity matrix among the grades to visualize the measure of degree of similarity across the clusters.\par
We form the dissimilarity matrix based on the Volume Weighted Valuations and use the metric: $$\text{Dissimilarity} (d)=2(1-\rho^2)$$ where $\rho$ is the product moment correlation correlation coefficient between these time series of Volume Weighted median Valuations for different grades. Thus we create the dissimilarity matrices between these grades, and obtain Figure~\ref{tea_dis}, where the darker shade represents a larger value of dissimilarity, while a lighter shade shows that the clusters are similar.

\begin{figure} 
\centering
\includegraphics[width=0.6\textwidth]{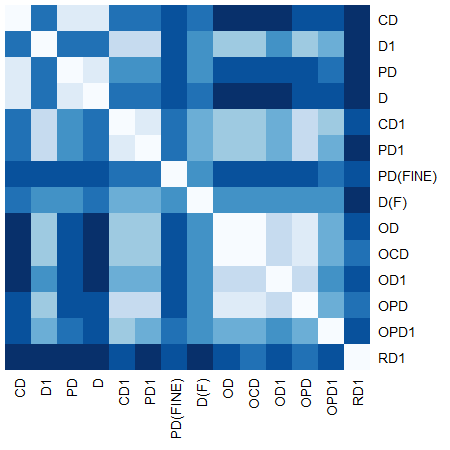}
\caption{Mosaic plot for dissimilarity based on correlation across clusters (Darker shades of blue indicates higher degree of dissimilarity)}
\label{tea_dis}
\end{figure}

\par One approach of clustering these is using hierarchical clustering based on this dissimilarity measure, to obtain the following clustering dendrogram in Figure \ref{dendo_dis}.
\begin{figure}
    \centering
    \includegraphics[width=\textwidth,trim={0 3cm 0 0},clip]{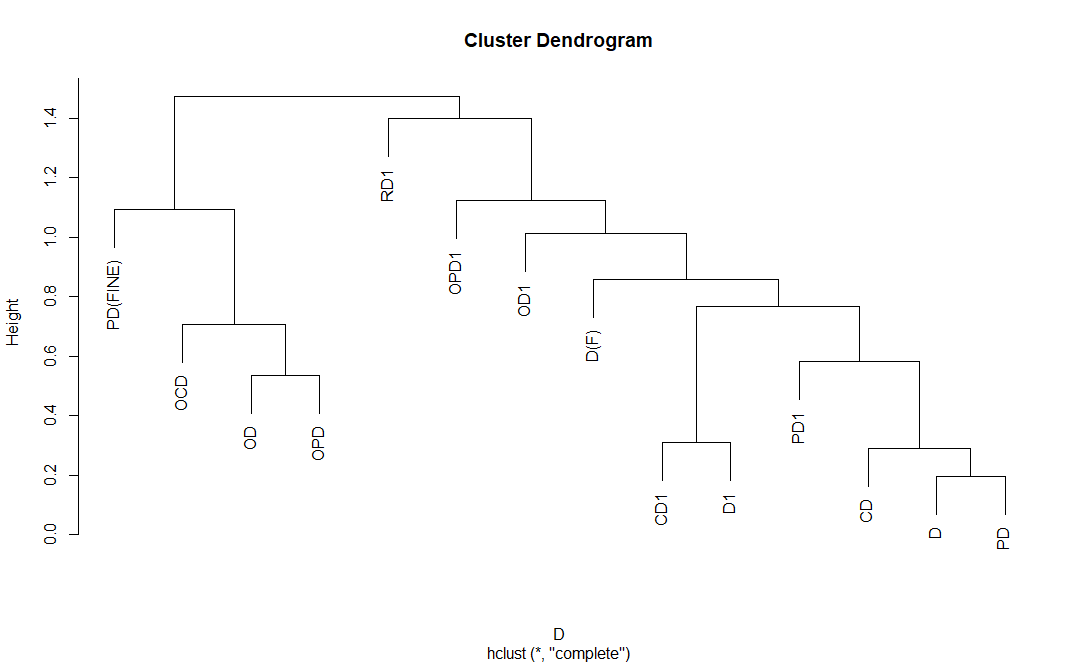}
    \caption{dendrogram for Volume Weighted Median Valuations based on 1-$\rho^2$}
    \label{dendo_dis}
\end{figure}

\par Analogously, considering Volume weighted mean valuations, we obtain Figure \ref{dendo_vol_dis}.

\begin{figure}
    \centering
    \includegraphics[width=\textwidth,trim={0 3cm 0 0},clip]{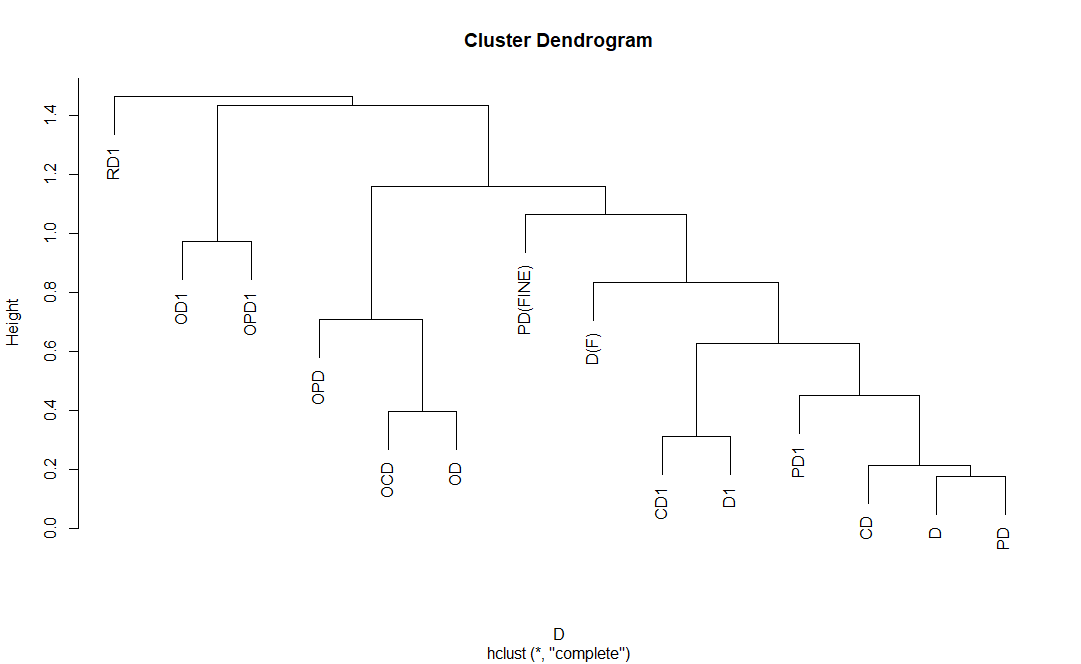}
    \caption{dendrogram for Volume Weighted Mean Valuation based on 1-$\rho^2$}
    \label{dendo_vol_dis}
\end{figure}

However, note that correlation is invariant to change of scale and origin, hence if a particular grade has even twice a valuation than another, using correlation as a measure of clustering would essentially nullify the effect of even twice the valuation, and would land the two grades into the same cluster, in spite of the fact that the grades are quite distinct in their market characteristics. \par 

Furthermore, we need to incorporate the idea that we have the data for multiple weeks, and the clustering should involve the clubbing for all the weeks combined.\par 

Thus we use the following idea: 
\begin{itemize}
    \item First we use \textbf{Bayesian Information Criteria} to figure out the number of clusters so that the model has the largest information. This figure came out to be $6$.
    \item Then, we used a Gaussian Mixture Model and ran the Expectation-Maximization Algorithm (often known in the literature as the EM-GMM clustering). We performed clustering both based on median as well as the mean, these two resulting in slightly different clusterings. We stick with the median based clustering due to its robustness.
    \item Both the price and the valuation was considered when applying EM-GMM clustering, in order to effectively capture the whole pattern present in the data.
\end{itemize}

Then, we define a new similarity structure, where the $i$-th and the $j$-th grade's similarity is proportional to the number of weeks they have occurred in the same cluster by EM-GMM method. Thus, we form a \textbf{similarity matrix}
with $(i,j)$-th entry of the matrix characterizing the similarity in terms of number of co-occurrence weeks. The Mosaic plot corresponding to these similarity matrices is given in Figure ~\ref{time_Mosaic_mean} and Figure~\ref{time_Mosaic_median}. Finally based on this similarity matrix, we conduct a hierarchical clustering to obtain the clusters, shown in Figure ~\ref{Median_cluster} and Figure ~\ref{Mean_cluster}.

\begin{figure}
    \centering
    \includegraphics[width=0.75\textwidth]{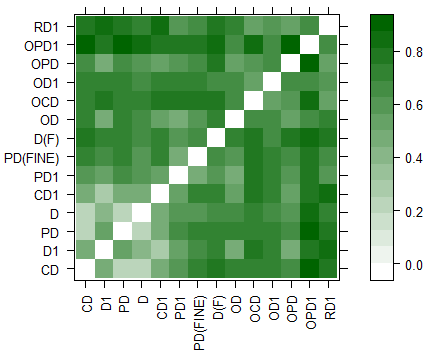}
    \caption{Mosaic plot for similarity matrix based on EM-GMM using Volume weighted Mean Price \& Valuations (Darker shades of green indicates higher degree of dissimilarity))}
    \label{time_Mosaic_mean}
\end{figure}

\begin{figure}
    \centering
    \includegraphics[width=0.75\textwidth]{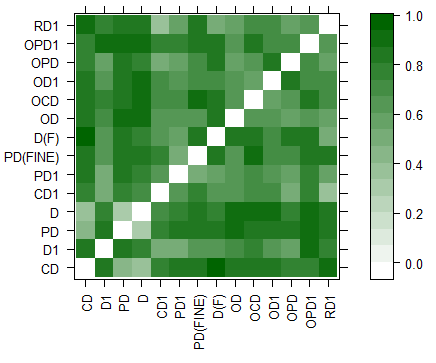}
    \caption{Mosaic plot for similarity matrix based on EM-GMM using Volume weighted Median Price \& Valuations (Darker shades of green indicates higher degree of dissimilarity)}
    \label{time_Mosaic_median}
\end{figure}

\begin{figure}
    \centering
    \includegraphics[width=\textwidth]{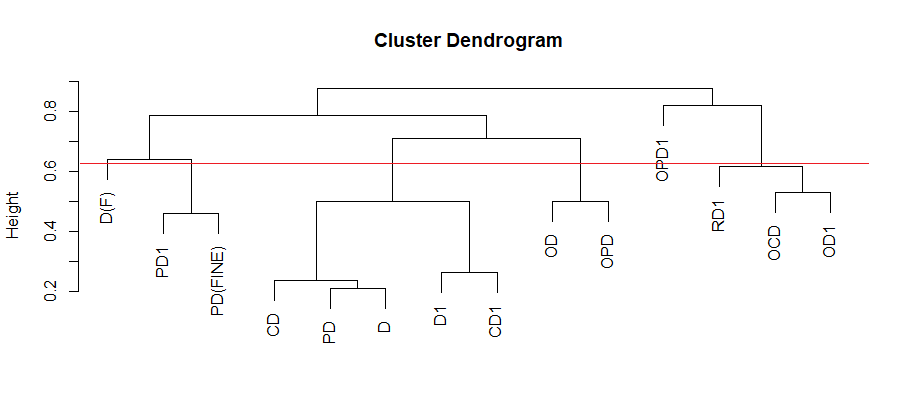}
    \caption{dendrogram for tea grades for volume weighted means by EM-GMM}
    \label{Mean_cluster}
\end{figure}

\begin{figure}
    \centering
    \includegraphics[width=\textwidth]{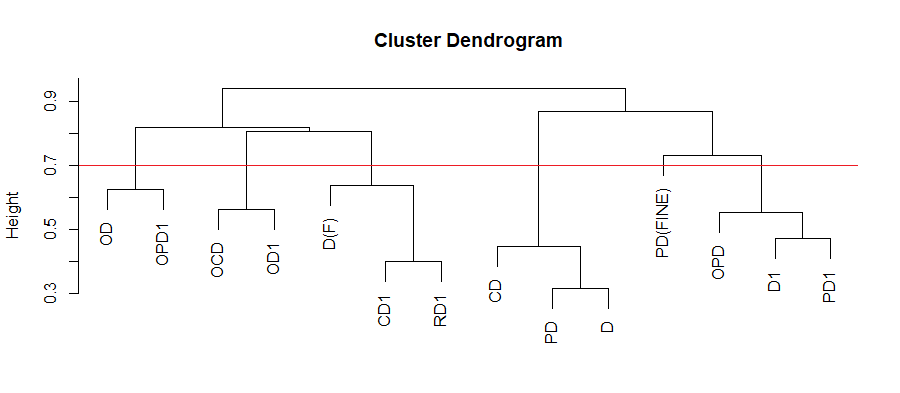}
    \caption{dendrogram for tea grades for volume weighted medians by EM-GMM}
    \label{Median_cluster}
\end{figure}

 Thus we obtain the following \textbf{6 clusters based on grade}:
\begin{itemize}
    \item \textbf{Cluster 1}: OD, OD-Special, OPD1 (3 grades)
    \item \textbf{Cluster 2}: OCD, OCD1, OD1 (3 grades)
    \item \textbf{Cluster 3}: D-Fine, CD1, CHD1, RD1 (4 grades)
    \item \textbf{Cluster 4}: D, D-Special, CD, CHD, CHU, PD, PD-Special (7 grades)
    \item \textbf{Cluster 5}: PD-Fine (1 grade)
    \item \textbf{Cluster 6}: OPD, OPD-Clonal, ORD, D1, D1-Special, PD1, PD1-Special (7 grades)
\end{itemize}
while the GT Dust category has been left out due to lack of sufficient data points.

\subsection{Diagnostics of Grade Clustering}
As mentioned before, according to Wikipedia \cite{wiki}, the three main categories for dust type tea leaf are \textbf{D1}, \textbf{PD} and \textbf{PD1}. However, our clustering algorithm shows that Tea Grade \textbf{D1} and \textbf{PD1} are similar in market characteristics by putting them into same cluster. On this note, we consider two different clusterings which might be possible.

\begin{enumerate}
    \item \textbf{Method 1:} The 6 clusters as mentioned above.
    \item \textbf{Method 2:} 8 clusters, cluster 1, 2, 3, 5 remaining as it is, while cluster 3 gets broken into two separate clusters, one containing grades PD, PD-Special and another containing CD, CHD, CHU, D, D-Special. Similarly, we divide cluster 6 into two separate clusters, one containing OPD, OPD-Clonal, ORD, D1, D1-Special and another containing PD1, PD1-Special.
\end{enumerate}

To find out which one of the above clusterings would be better for further analysis, we find out the proportion of total variation of both weekly price and weekly valuation which is explained by the clusters. Some of the obtained results for both of the clusterings are given in Table \ref{var-clusters}. It was found that, for both weekly price and weekly valuation, about 50\% variation of these variables over different lots are explained by the clusters alone. Also, using 8 clusters instead of 6 clusters increases this proportion of explained variation by at most 2\%, which is not significant in contrast to the loss of simplicity of subsequent works. This diagnostic suggests us to stick with the original clustering, with 6 clusters as defined previously.

\begin{table}
\centering
\caption{Proportion of Explained Variation of Weekly Price and Valuation by Clusters}
\label{var-clusters}
\begin{tabular}{|c|c|c|c|c|}
\hline
\multirow{3}{*}{Week} & \multicolumn{4}{c|}{Explained Proportion of Variation} \\
\cline{2-5}
& \multicolumn{2}{c|}{For Valuation} &\multicolumn{2}{c|}{For Price} \\
\cline {2-5} 
& with 6 clusters & with 8 clusters & with 6 clusters & with 8 clusters \\
\hline 
2    & 54.33\%  & 54.81\%  & 44.1\%  & 44.8\%   \\ \hline
3    & 56.79\%  & 57.28\%   & 44.2\%  & 44.5\%   \\ \hline
4    & 58.26\%  & 58.77\%  & 48.2\%  & 50.1\%   \\ \hline
\multicolumn{5}{|c|}{\vdots}    \\ \hline
45   & 63.93\%    & 65.27\%     & 60.2\%   & 61.4\% \\ \hline
46   & 56.84\%    & 57.09\%  & 51.4\%  & 51.6\%  \\ \hline
\end{tabular}
\end{table}

\subsection{Clustering by Source}
Now, note that, if we just cluster the tea based on their grades, then we are foregoing valuable information about the source from which the tea has been produced. This information would be very relevant for our subsequent discussions, and it would not be a good idea to get rid of it. Valuations about a tea grade depend on gardens they are originating from, as it utilizes the idea about the environment used for their nourishment, soil levels, and other significant factors. Hence the source of the tea dust is of utmost importance.\par
However, there are 238 tea gardens (293 including their Clonal, Royal, Gold and Special variants) from which the tea has originated, and again as above, we suspect that maintaining track of each of the tea gardens would be intractable as well as redundant, as the tea also show similarity in characteristics, a significant factor of which might be geographical proximity. Hence we undergo clustering based on the volume-weighted median to obtain dendrogram given in Figure ~\ref{source_medians}. These dendrogram has again been clustered using the EM-GMM algorithm and then the time-based similarity matrix as in the preceding section. \par 

\begin{figure}
    \centering
    \includegraphics[width=\textwidth]{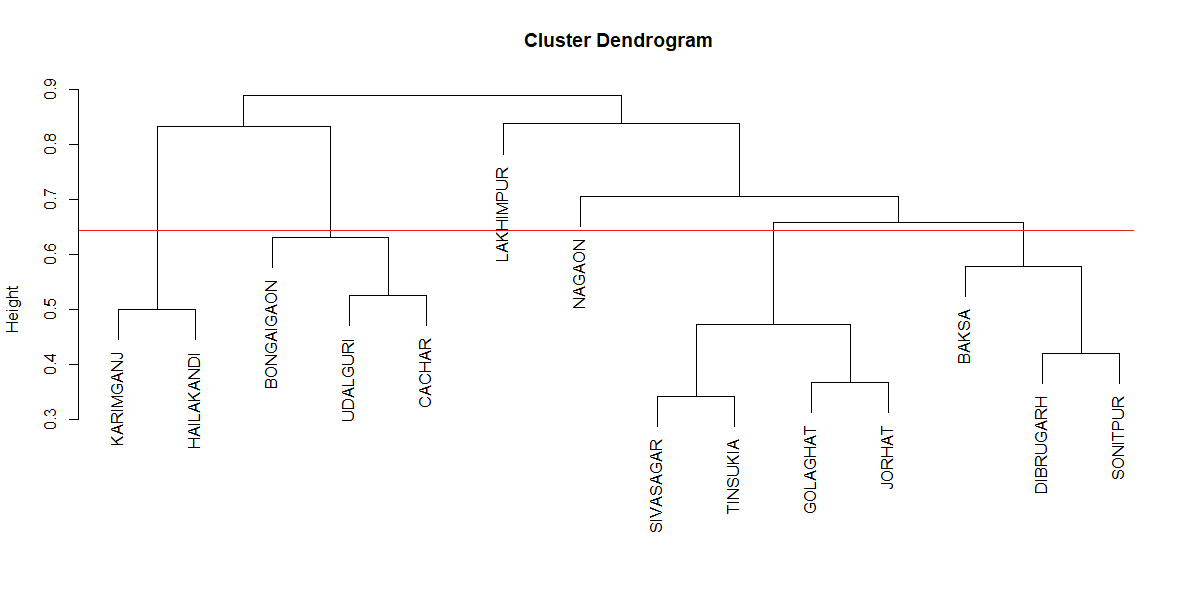}
    \caption{dendrogram for volume-weighted medians by source}
    \label{source_medians}
\end{figure}

\begin{figure}[ht]
    \centering
    \includegraphics[width=0.6\textwidth]{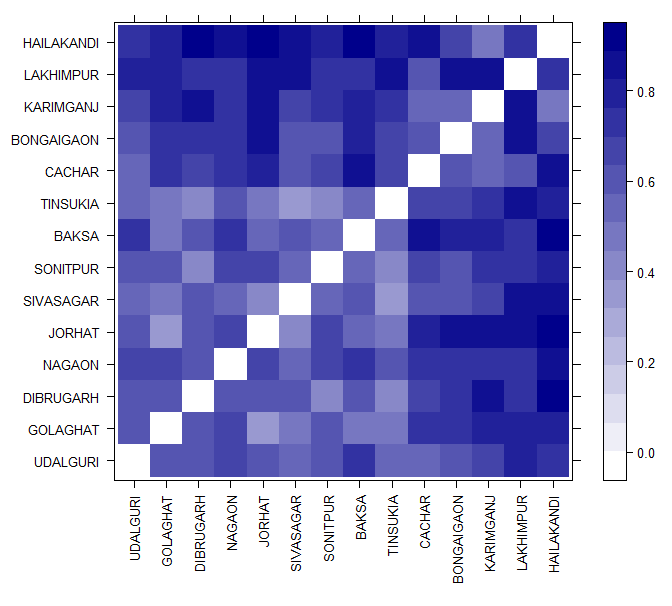}
    \caption{Mosaic plot for dissimilarity based on EM-GMM clustering for source (Darker shades for more dissimilarity)}
    \label{mosaic_source_EM}
\end{figure}

Before going into the final source-based clusterings, we need to keep the following things in mind as well: 
\begin{itemize}
    \item The districts for the source of the tea dust vary from the northern fringes of West Bengal to the entirety of Assam. Hence it might be a good idea to keep the districts of West Bengal and Assam well segregated.
    \item Geographical proximity might be a concern for the explanation of the districts that belong to the same cluster.
    \item Topography, soil structure, and water source may also be the reason for similar or dissimilar market characteristics of the tea dust.
    \item Most of the tea dust from West Bengal comes from Jalpaiguri, as portrayed in the data, while the other districts have a significantly lesser number of such packets to be sold. Hence it would be good if the districts of West Bengal are classified keeping this in mind. 
\end{itemize}
We attach here the maps of the tea-dust producing districts of West Bengal and Assam to give a visualization of the geographical proximities (in Figure ~\ref{wb_map} and Figure ~\ref{assam_map} ).\par 

\begin{figure}
    \centering
    \includegraphics[width=0.99\textwidth,trim={0.1cm 26cm 0.1cm 0.1cm},clip]{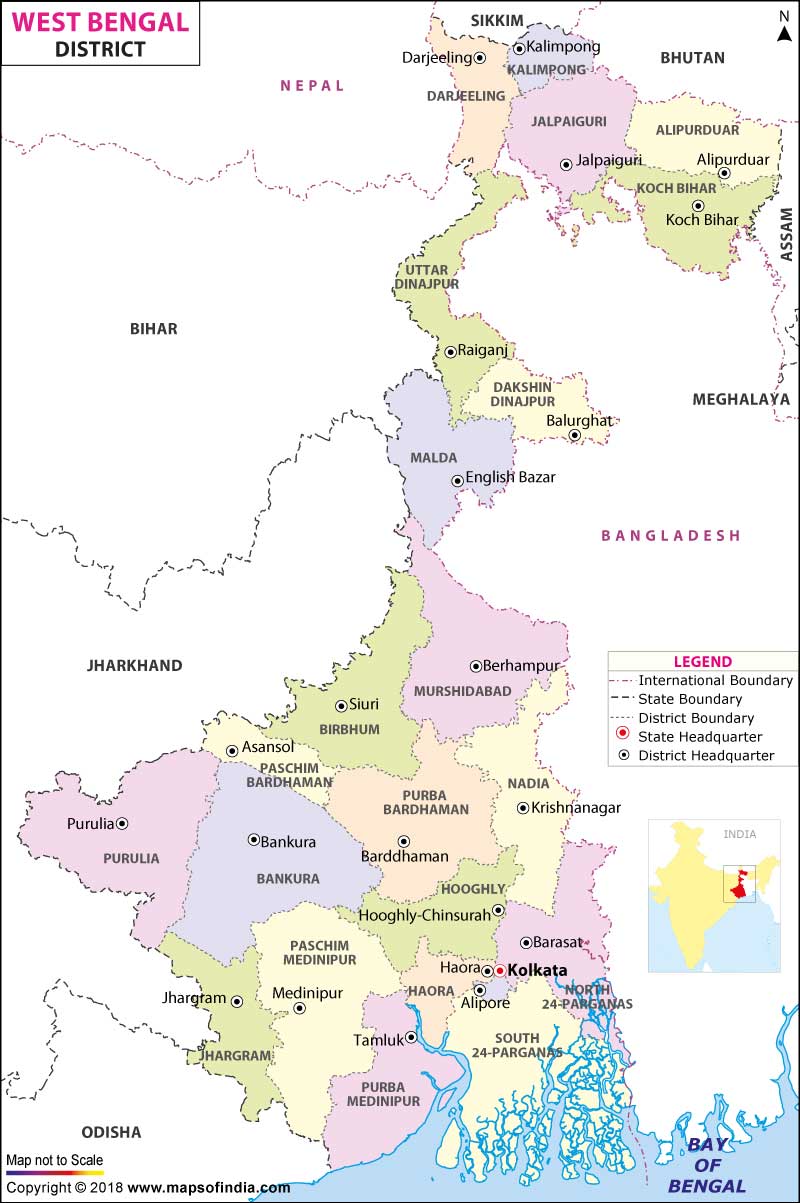}
    \caption{Map of North of West Bengal \cite{wb_map}}
    \label{wb_map}
\end{figure}

\begin{figure}
    \centering
    \includegraphics[width=0.99\textwidth,trim={0.15cm 0.65cm 0.15cm 0.65cm},clip]{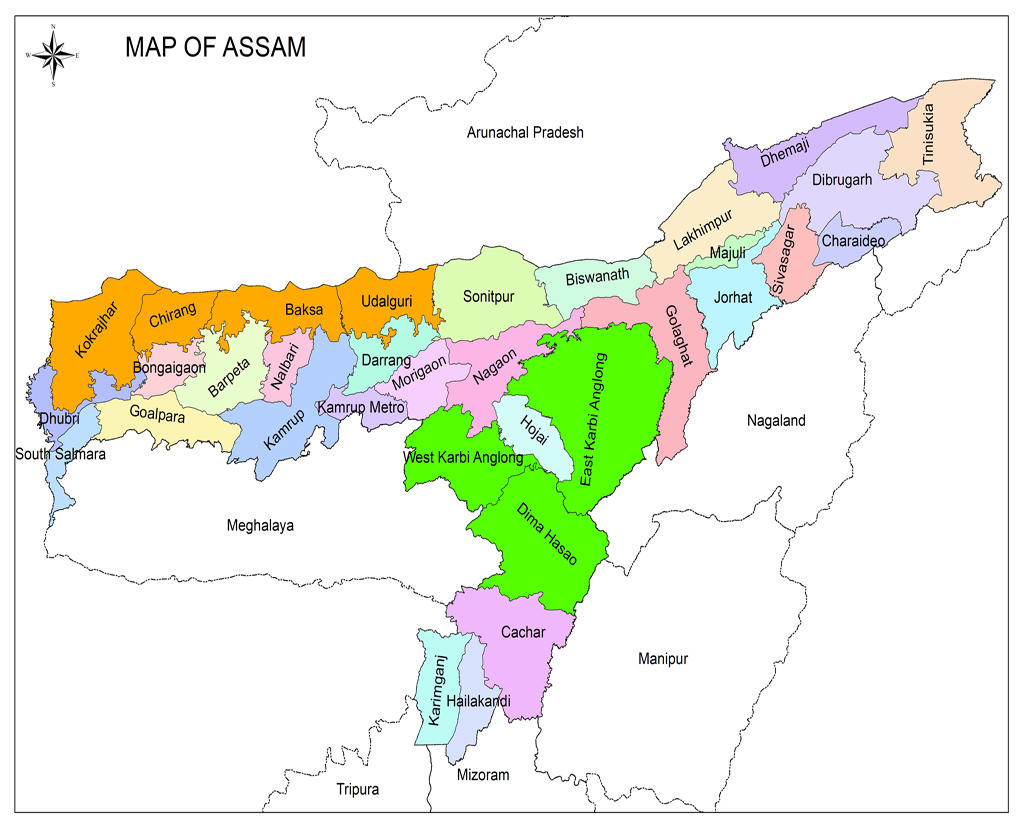}
    \caption{Map of Assam \cite{assam_map}}
    \label{assam_map}
\end{figure}

Thus we finally use the following 7 clusters based on the source of the tea dust.
\begin{itemize}
    \item \textbf{Cluster 1}: Darjeeling, Cooch Behar (Koch Behar), Uttar Dinajpur, Jalpaiguri (4 West Bengal districts)
    \item \textbf{Cluster 2}: Karimganj, Hailakandi 
    \item \textbf{Cluster 3}: Bongaigaon, Cachar, Udalguri, Darrang, Dima Hasao 
    \item \textbf{Cluster 4}: Lakhimpur
    \item \textbf{Cluster 5}: Nagaon
    \item \textbf{Cluster 6}: Sivasagar, Tinsukia, Golaghat, Jorhat
    \item \textbf{Cluster 7}: Baksa, Dibrugarh, Sonitpur
\end{itemize}

\section{Distribution of Volume over Different Months}

Once the clustering for tea grades and sources are obtained, the next idea would be to create a temporal clustering. The data for 2018, presented in the form of weeks, are  put into buckets of months. This is done to capture the seasonal variations in the market characteristics of tea dust. But again the exact week number might be too redundant an information, as the tea dust appearing in the market does not fluctuate as frequently as weeks, but might vary by seasons. Thus, to incorporate such a possibility of market dynamics, we calibrate the data of tea dust grades by months. This gives us the Table ~\ref{tea_months}. To obtain Table ~\ref{tea_months}, we find out the number of tea packets offered in a lot and multiply it with the net average weight of the tea packets to obtain the total amount (or volume) of tea offered. Then, for each cluster of grade, we find the proportion of its total volume which is offered during a specified month. This gives us a basis for clustering to analyze the supply side of market dynamics.

\begin{table}
    \centering
    \begin{tabular}{|c|c|c|c|c|c|c|}
    \hline 
    Clusters$\rightarrow$     &1  &2 &3 &4 &5 &6 \\
    \cline{2-7} Months $\downarrow$ & & & & & &\\
    \cline{1-1} 
    Jan &   0.099  &  0.134 &  0.105  &  0.056  &  0.038 &   0.100\\
    Feb &  0.070  &  0.125 &  0.082  &  0.037  &  0.000 &   0.085\\
    Mar &   0.008  &  0.018 &  0.019  &  0.009  &  0.000 &   0.016\\
    Apr &   0.000  &  0.000 &  0.000  &  0.000	&  0.000 &   0.000\\
    May &   0.027  &  0.010 &  0.008  &  0.031  &  0.018 &   0.015\\
    Jun &   0.108  &  0.052 &  0.069  &  0.077  &  0.053 &   0.059\\
    Jul &  0.129  &  0.091 &  0.100  &  0.125  &  0.025 &   0.098\\
    Aug &   0.115  &  0.109 &  0.114  &  0.129  &  0.109 &   0.115\\
    Sep &   0.139  &  0.155 &  0.123  &  0.138  &  0.177 &   0.128\\
    Oct &   0.051  &  0.062 &  0.089  &  0.090  &  0.184 &   0.096\\
    Nov &   0.084  &  0.086 &  0.132  &  0.129  &  0.191 &   0.126\\
    Dec &  0.171  &  0.156 &  0.160 & 0.179 &  0.205 &   0.162\\
    \hline 
    \end{tabular}
    \caption{Tea grade proportions by months \newline Figures denote ratio of volume of each cluster grade occurring in that month and the total volume of all packets of tea grade of that cluster in the dataset }
    \label{tea_months}
\end{table}

The mosaic plot for the above table has been included in Figure ~\ref{week_mosaic} for a better visualization.

\begin{figure}
    \centering
    \includegraphics[width=\textwidth]{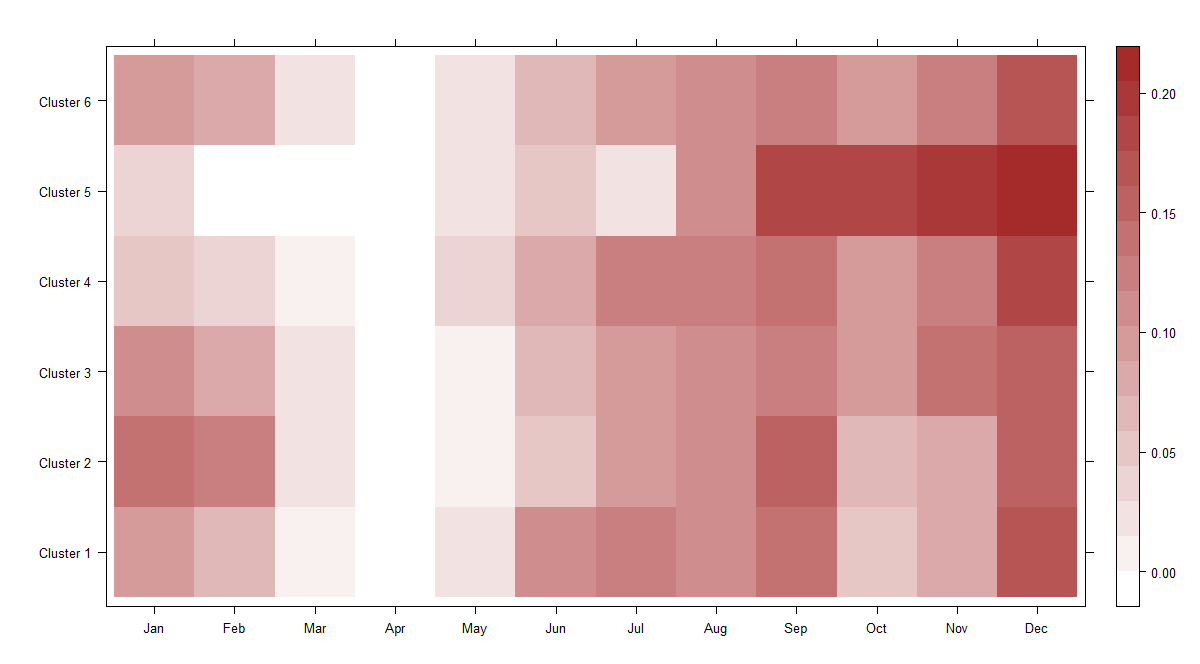}
    \caption{Mosaic plot for proportion of Volume of tea occurring for the grade clusters across months (Darker shades suggest larger proportions)}
    \label{week_mosaic}
\end{figure}

\section{Predicting Salability of Offered Tea lots}

For analyzing the auction of tea packets, we should concern ourselves with the proportion of tea packets to be sold, and relate its valuation and several other characteristics to it. A successful attempt at predicting the probability of being sold (or being unsold) of an incoming tea packet based on its Grade, Source, Valuation and the current month, would give an insight for automating the auction process, alongside enabling an opportunity to study the effect of Valuation on determining the market characteristics of tea grades.

\subsection{A Primary Inspection}
Primary inspection is made to see whether any particular type of tea grades are more likely to be sold at the auction than other types. Table ~\ref{sold-prop-grade} shows the proportion of tea lots being sold and the total number of tea lots offered across different grades.

\begin{table}
\centering
\caption{Proportion of Tea Lots being Sold across different Grades}
\label{sold-prop-grade}
\begin{tabular}{|c|c|c|c|c|c|}
\hline
Grade  & \begin{tabular}[c]{@{}c@{}}Total Lots\\ Offered\end{tabular} & \begin{tabular}[c]{@{}c@{}}Proportion \\ of Sold\end{tabular} & Grade         & \begin{tabular}[c]{@{}c@{}}Total Lots\\ Offered\end{tabular} & \begin{tabular}[c]{@{}c@{}}Proportion\\ of Sold\end{tabular} \\ \hline
CD & 1720  & 0.768 & OD (Special)  & 3   & 0.333 \\ \hline
CD1 & 1400  & 0.819 & OD1   & 125 & 0.896 \\ \hline
CHD & 5 & 0.8 & OPD & 1064 & 0.822 \\ \hline
CHD1 & 23 & 0.696 & OPD (Clonal)  & 3  & 1  \\ \hline
CHU & 5 & 1 & OPD1  & 61 & 0.754 \\ \hline
D & 7932 & 0.816 & ORD  & 8 & 1 \\ \hline
D (Fine) & 216 & 0.842 & PD & 7119 & 0.771 \\ \hline
D (Special)  & 6 & 1 & PD (Fine) & 111 & 0.945 \\ \hline
D1 & 3477 & 0.842 & PD (Special)  & 23 & 0.609 \\ \hline
D1 (Special) & 1 & 1 & PD1  & 500 & 0.872 \\ \hline
OCD & 206 & 0.888 & PD1 (Special) & 5 & 0.8 \\ \hline
OD & 2229 & 0.796 & RD1  & 77 & 0.987 \\ \hline
\end{tabular}
\end{table}

Note that, Table ~\ref{sold-prop-grade} shows that Fine variant of a tea grade is offered more rarely than its original variant, and its probability of getting sold at the auction also increases. Also, the Special variant of any tea grade is rarer to be offered than its Fine variant, and for this reason, there is not a significant number of observations to conclude whether it increases or decreases the probability of getting sold at the auction.

To find out whether the probability of being sold significantly depends on the variant of tea grades, we perform a simple one-way Analysis of Variance model with the indicator of being sold as the response variable and the type of variant as our treatment variable. We obtain the results as shown in Table \ref{anova-tea-grade}. Clearly, the probability of being sold for Fine variants of tea grades are higher than regular variant, and the p-value is small indicating that there are a significant number of observations to support this. However, the proportion of being sold is possibly lower in Special variants than in Regular ones, but higher p-value indicates that there is not enough evidence to support this.

\begin{table}
\centering
\caption{Output for Analysis of Variance of Sellablity on variant of Tea Grades}
\label{anova-tea-grade}
\begin{tabular}{|c|c|c|c|c|}
\hline
Coefficients & Estimate & Std. Error & t statistic & p value \\ \hline
Regular (Intercept) & 0.808    & 0.002721   & 297.20     &  $2\times 10^{-16}$  \\ \hline
Fine                & 0.077    & 0.025131   & 3.10        & 0.00194 \\ \hline
Special             & -0.098   & 0.063752   & -1.54       & 0.12347 \\ \hline
\end{tabular}
\end{table}

Similar to this, we tried to find out whether different variants of Source (for example, Clonal, Gold, Royal, etc.) affect the probability of getting sold. Again we perform an Analysis of Variance model, however, with the variant of the tea garden (or Source) as our treatment effect. We obtain the results as shown in Table ~\ref{anova-tea-source}. From this, we note that if the tea packet has come from a Clonal tea garden, its selling probability is expected to be higher than Regular ones by 0.044, and the smaller value of p-value indicates evidence to support this claim. Similarly, the tea packets produced from the Gold type variant of Garden is expected to be 15\% less probable to be sold at the auction. Also, with 95\% confidence, we can say that the tea packets produced from the Royal type variant of Garden are 22\% less likely to be sold.

\begin{table}
\centering
\caption{Output for Analysis of Variance of Sellablity on variant of Tea Gardens}
\label{anova-tea-source}
\begin{tabular}{|c|c|c|c|c|}
\hline
Coefficients        & Estimate & Std. Error & t statistic & p value \\ \hline
Regular (Intercept) & 0.808    & 0.002881   & 280.653     & $2\times 10^{-16}$  \\ \hline
Clonal              & 0.044    & 0.010272   & 4.259       & $2.06\times 10^{-5}$ \\ \hline
Gold                & -0.151   & 0.032931   & -4.588      & $4.05\times 10^{-6}$ \\ \hline
Royal               & -0.225   & 0.113282   & -1.987      & 0.0469  \\ \hline
Special             & -0.027   & 0.013983   & -1.96       & 0.05    \\ \hline
\end{tabular}
\end{table}

\subsection{Model for Prediction}

To build a predictive model to predict whether a tea packet will be sold at the auction or not, based on its Valuation, Grade, Source, and the current time, we use three competing models.

\begin{enumerate}
    \item Logistic Regression
    \item Generalized Additive Model
    \item Mixture of Logistic Regression
\end{enumerate}

We divide the total dataset of the year 2018 into training and cross-validation sets, with the training set containing 70\% of the samples. The cross-validation set is used to select the model. The dataset of the year 2019 is kept as a testing set, which is used to evaluate the performance of the finally selected model for prediction. The percentage of sold tea lots is kept almost similar for both the training and testing sets about 81\%. For each of the sets, we consider each combination of Valuation, Grade, Source, and Month, and obtain the proportion of sold tea lots among all tea lots offered under that combination. On the other hand, the predicted probability of being sold under that combination is estimated from the trained model. Both of these probabilities are visually and analytically compared against each other to assess the performance of the trained model for both sets. For analytical comparison, we use three measures as follows;

\begin{enumerate}
    \item \textbf{Null and Residual Deviance}: For a generalized linear model, $$\underbrace{\text{Null deviance}}_{\text{df}}=\underbrace{2\left(\Lambda(\text{Saturated Model})-\Lambda(\text{Null Model})\right)}_{\text{df(Saturated Model) - df(Null Model) }}$$ and $$\underbrace{\text{Residual deviance}}_{\text{df}}=\underbrace{2\left(\Lambda(\text{Saturated Model})-\Lambda(\text{Proposed Model})\right)}_{\text{df(Saturated Model) - df(Proposed Model) }}$$
    where $\Lambda(\cdot)$ stands for the log-likelihood under the model in the argument. \par
    The saturated model is characterized by each data point bringing in its new parameter, ie, it provides no further scope for the addition of parameters corresponding to the data points. Thus all $n$ parameters are to be estimated. \par The null model is the exact antipodal in the sense that it assumes a common parameter for all the data points, and thus only one parameter needs to be estimated. The proposed model is the model one proposes, with $p+1$ parameters.\par 
    A small null or residual deviance indicates that the corresponding one or p+1 parameter model explains the data well. Formally, under a truly good fit, the deviances follow a chi-squared distribution with the mentioned degrees of freedom.
    \item \textbf{RMSE:} Root Mean Squared Error measures the $L_2$ distance between the predicted probabilities and the true probabilities. Hence smaller values are preferred.
    \item \textbf{MAE:} Mean Absolute Error measures the $L_1$ distance between the predicted and true probabilities. Hence smaller values are preferred.
\end{enumerate}

\subsection{Logistic Regression}

We began with attempting to fit a logistic regression model with the full data. The following results came up:
\begin{itemize}
    \item Null deviance: 18055 on 18381 degrees of freedom.
    \item Residual deviance: 17813  on 18359  degrees of freedom.
    \item Null deviance - residual deviance = 242 with degrees of freedom 22.
    \item Pseudo $R^2$: 0.01340349.
    \item RMSE is 0.295 for training set and 0.312 for cross-validation set.
    \item MAE is 0.214 for training set and 0.235 for cross-validation set.
    \item Figure ~\ref{logit_fit} shows how bad the fit is for training and testing sets respectively.
\end{itemize}

\begin{figure}
    \centering
    \includegraphics[width=\textwidth]{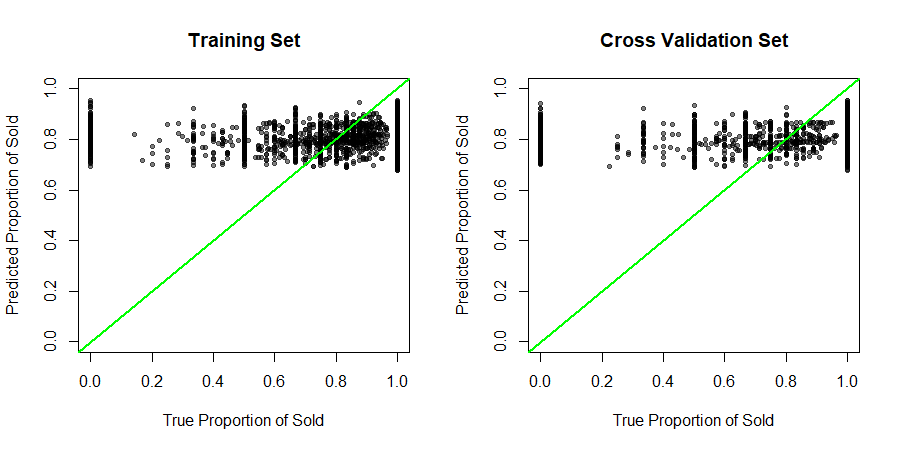}
    \caption{Fit of logistic regression for predicting sold grades for training and testing sets}
    \label{logit_fit}
\end{figure}

Hence this model fails miserably in predicting whether a packet would be sold.

\subsection{Generalized Additive Model}

We use a Generalized Additive Model with a binomial family, with a smooth cubic spline fitted on the Valuation of tea packets as a predictor. The following results came up:
\begin{itemize}
    \item Null deviance: 18055 on 18381 degrees of freedom.
    \item Residual deviance: 17627  on 18357  degrees of freedom.
    \item Null deviance - residual deviance = 429 with degrees of freedom 24.
    \item Pseudo $R^2$: 0.02373118.
    \item RMSE is 0.287 for the training set and 0.306 for the cross-validation set.
    \item MAE is 0.21 for the training set and 0.2322 for the cross-validation set.
    \item Figure ~\ref{gam_fit} shows some improvement over logistic regression, but still, the fit remains too bad to be of any use.
\end{itemize}

\begin{figure}
    \centering
    \includegraphics[width=\textwidth]{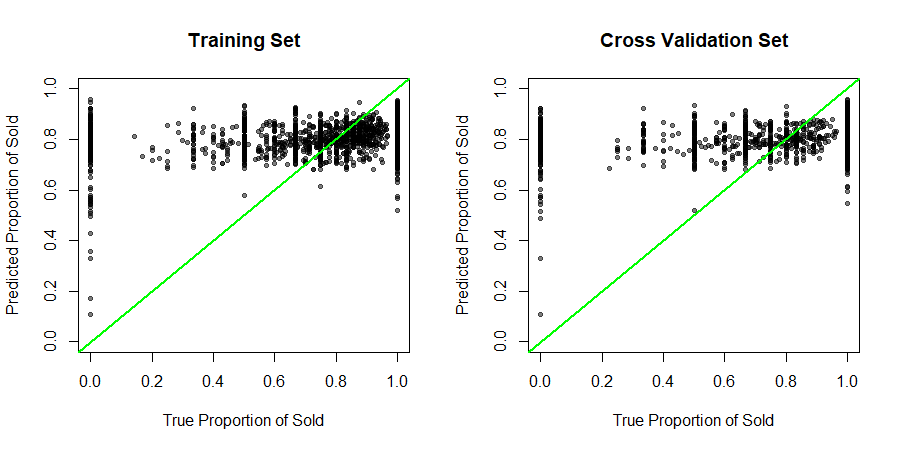}
    \caption{Fit of Generalized Additive Model for predicting sold grades for training and testing sets}
    \label{gam_fit}
\end{figure}

\subsection{Mixture of Logistic Regression}
We try using a mixture of logistic regressions to predict the selling potential of the packets. The model, in general, for a mixture of $S$ components, is given by \cite{logit_mix} (and extended in \cite{logit_mix_2}, \cite{logit_mix_3})
\[ H(y|T,\mathbf{x},\mathbf{w},\mathbf{\Theta})=\sum_{s=1}^S \pi_s(\mathbf{w},\mathbf{\alpha})\text{Bi}(y|T,\theta_s(x)) \]
where $\mathbf{w}$ stands for the concomitant variables on which the mixing proportions $\pi_s$ depend, $\text{Bi}(y|T,\theta_s(\textbf{x}))$ is the binomial distribution with number of trials equal to $T$ and success probability $\theta_s\in(0,1)$, is modelled by usual logistic modelling; $\text{logit} \left(\theta_s(\mathbf{x})\right)=\mathbf{x}^\top \mathbf{\beta}^s$. The concomitant variable is assumed to have a multinomial logit model, i.e. of the form $$\pi_s(\mathbf{w},\mathbf{\alpha})=\dfrac{e^{\mathbf{w}^\top\mathbf{\alpha_s}}}{\sum_{u=1}^S e^{\mathbf{w}^\top\mathbf{\alpha_u}}} \ \forall s$$

Here, we have used our independent variable $x$ to be the corresponding valuations, $T$ being the number of packets arriving, and success denoting the event that the packet is sold. The concomitant variables are the source, grade cluster and month of the packets.\par 

Here we consider 2 to 5 component mixture for this. Based on the Bayesian Information Criterion, the 3 component mixture of logistic regression is chosen which yields a BIC value 6183.014.

\begin{itemize}
    \item RMSE is 0.222 for the training set and 0.2488 for the cross-validation set.
    \item MAE is 0.1522 for the training set and 0.1811 for the cross-validation set.
    \item Figure ~\ref{mix_logit_fit} provides the plots for fits, which shows significant improvement over the previous models. Figure ~\ref{rootogram} gives the rootogram of the components, and the peaks at 0 and 1 suggest how well these components can be identified. Note that, the observations assigned to component 1 are marked in the rootogram. The peak at 1 for component 1 shows that this component is well separated from others.
    \item Taking into consideration the metrics for fit as well as the plots, we finalize this model to predict the probability of a particular packet of tea being sold, given its corresponding covariates.
    \item As we decided to use this as our final prediction model, we can update its parameters based on all of the training and cross-validation set. Finally, we evaluate its performance on the testing set (comprises of data in 2019).
    \begin{itemize}
        \item Updated model achieves RMSE of 0.214 and 0.2519 in the whole training set and testing set respectively.
        \item This achieves an MAE of 0.146 and 0.1684 in the whole training set and testing set respectively.
        \item Figure ~\ref{mix_logit_fit_final} provides the plots for the fitted model for both training and testing sets.
        \item Table ~\ref{mix-logit-summary} provides the summary of this model.
    \end{itemize}
\end{itemize}

\begin{figure}
    \centering
    \includegraphics[width=\textwidth]{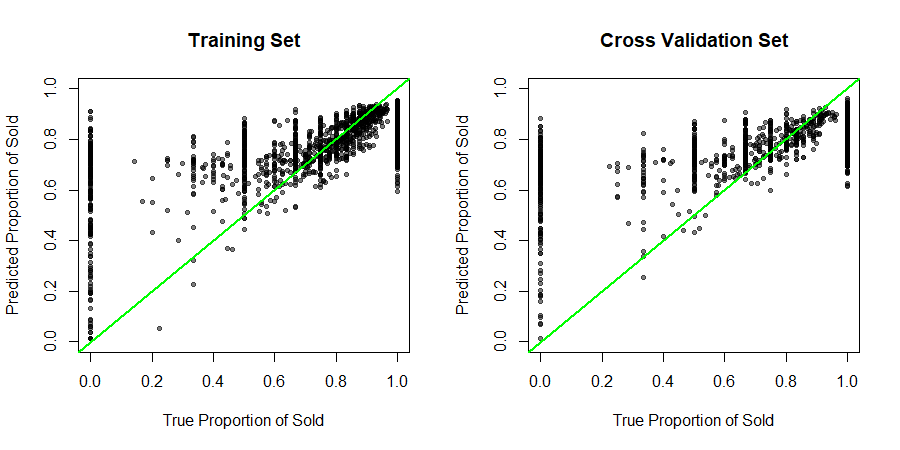}
    \caption{Fit of 3 component mixture of logistic regression for predicting sold grades for training and cross-validation sets}
    \label{mix_logit_fit}
\end{figure}

\begin{figure}
    \centering
    \includegraphics[width=\textwidth]{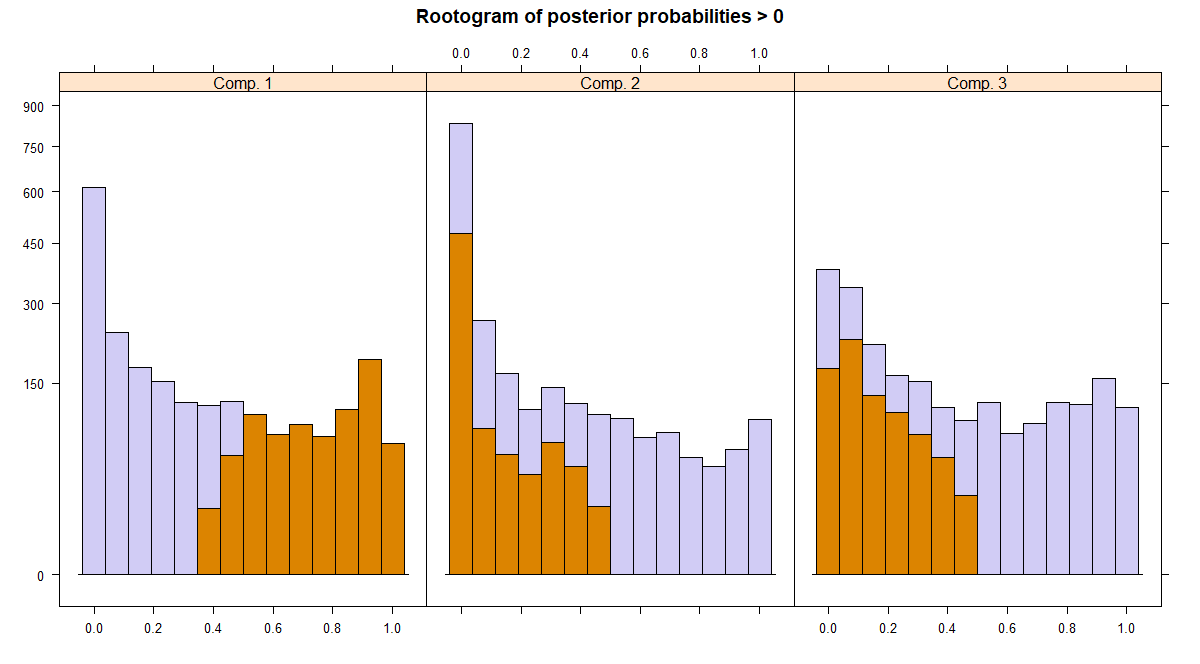}
    \caption{Rootogram of posterior probabilities obtained from fitted model (Component 1 is marked)}
    \label{rootogram}
\end{figure}

\begin{figure}
    \centering
    \includegraphics[width=\textwidth]{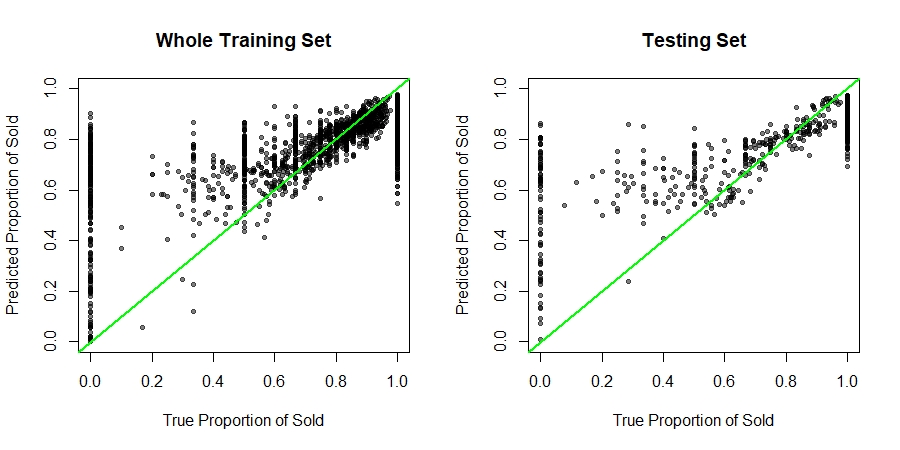}
    \caption{Updated Fitted model of 3 component mixture of logistic regression for predicting sold grades for whole training and training sets}
    \label{mix_logit_fit_final}
\end{figure}

\begin{table}
\centering
\caption {Summary of Mixture of Logistic Regression Fitted Model}
\label{mix-logit-summary}
\begin{tabular}{|c|c|c|c|c|c|c|}
\hline
\multirow{2}{*}{} & \multicolumn{2}{c|}{Component 1} & \multicolumn{2}{c|}{Component 2} & \multicolumn{2}{c|}{Component 3} \\ \cline{2-7} 
    & Intercept       & Valuation      & Intercept       & Valuation      & Intercept       & Valuation      \\ \hline
Estimate          & 0.5355          & 0.007586      & -1.48          & 0.0117       & 7.2017         & -0.03344        \\ \hline
Std. Error        & 0.1622          & 0.00092       & 0.1819          & 0.00092        & 0.3329         & 0.001882        \\ \hline
Z Value           & 3.3020          & 8.2702         & -8.1369         & 12.7798        & 21.635         & -17.772         \\ \hline
p value           & 0.0009           & 2.2e-16          & 4.05e-16         & 2.2e-16          & 2.2e-16           & 2.2e-16        \\ \hline
\end{tabular}
\end{table}

\section{Distribution of price to valuation ratio for grade clusters}
Valuation by experts does provide a significant knowledge of what the final price of the transaction would be. The entire process runs with the base price being set at some fixed proportion of the valuations, and thus the following price of transaction revolves significantly across this measure. In this section, we would attempt to fit distributions over the ratio of price and valuations to account for its variability and shape of distribution curves. We have attempted this exercise with the natural logarithm of the ratio of price and volume, to have full support over the real numbers. \par 
Observe that if, $\log\left(\dfrac{X}{Y}\right)\sim N(\mu,\sigma^2)$, then $$\mathbb{E}\left(\dfrac{X}{Y}\right)=\exp\left(\mu+\dfrac{\sigma^2}{2}\right)$$ and $$\Var\left(\dfrac{X}{Y}\right)=\exp\left(2\mu+\sigma^2\right)\left(e^{\sigma^2}-1\right)$$
The histograms of the ratio of price and valuations for various grade clusters as obtained from the data are given in Figure ~\ref{p/v_dist1} and Figure ~\ref{p/v_dist2}.

\begin{figure}
    \centering
    \includegraphics[width=\textwidth]{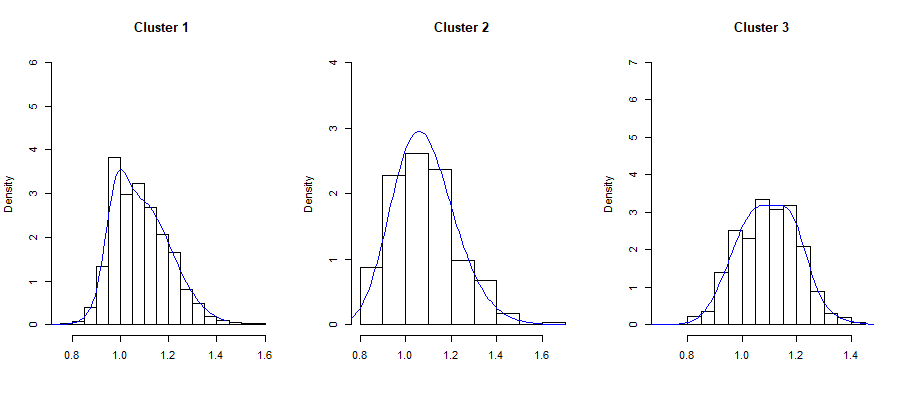}
    \caption{Histogram of Price/Value Ratio and Fitted Distributions for Clusters 1,2 and 3}
    \label{p/v_dist1}
\end{figure}

\begin{figure}
    \centering
    \includegraphics[width=\textwidth]{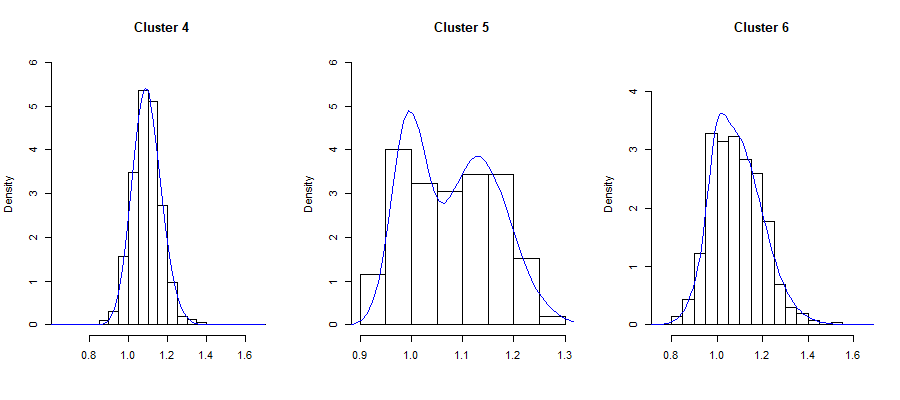}
    \caption{Histogram of Price/Value Ratio and Fitted Distributions for Clusters 4, 5 and 6}
    \label{p/v_dist2}
\end{figure}

\subsection{For Cluster 1 (OD, OD-Special, OPD1)}
We have attempted fitting a single normal distribution as suggested by the histogram. However, this results in a poor fitting of the data. Hence, we attempt to fit a mixture of two normal distributions to the data. The chi-squared goodness of fit yielded a p-value of 0.2274 and the Kolmogorov-Smirnov statistic yielded a p-value of 0.1588, which is quite satisfactory for our case. Thus we report the distribution of the ratio of price and valuation of Cluster 1 to be a \textbf{mixture of two log-normal distributions} with the following properties
\begin{center} 
\begin{tabular}{|c|c|c|c|c|}
    \hline
    Proportion & $\mu$ & $\sigma$ & Expectation & Variance\\
    \hline
    0.1909724 & -0.01608057 & 0.04471298 & 0.985 & 0.00194179\\
    0.8090276 & 0.09976462 & 0.10636915 & 1.1111792 & 0.01404943\\
    \hline 
\end{tabular}
\end{center}

\subsection{For Cluster 2 (OCD, OCD1, OD1)}
The histogram shows a single modal distribution. We tried fitting a single lognormal distribution, our p-values for  the chi-square goodness of fit statistic came out as 0.2665 and for One sample Kolmogorov Smirnov test came out as 0.6908, which suggests a reasonably good fit. Thus we report the ratio for this cluster to follow a \textbf{an unimodal log-normal distribution} with the following properties.
\begin{center}
    \begin{tabular}{|c|c|c|c|}
    \hline
    $\mu$ & $\sigma$ & Expectation & Variance\\
    \hline
    0.07285774 & 0.12655414 & 1.08422529 & 0.01897904 \\
    \hline
    \end{tabular}
\end{center}

\subsection{For Cluster 3 (D-Fine, CD1, CHD1, RD1)}
The same pattern follows as in Cluster 2, unimodal from histogram, but a single log-normal fits badly (p-value 2e-04). But fit with mixture of two log-normals give reasonably well fits (p-value 0.3855). Thus we report a \textbf{mixture of two log-normal distributions} with the following properties.
\begin{center}
    \begin{tabular}{|c|c|c|c|c|}
    \hline 
    Proportion & $\mu$ & $\sigma$ & Expectation & Variance\\
    \hline 
    0.8564837 & 0.07299101 & 0.10194329 & 1.081325 & 0.012214861 \\
     0.1435163 & 0.17518446 & 0.04294408 & 1.192565 & 0.002526254 \\
    \hline
    \end{tabular}
\end{center}

\subsection{For Cluster 4 (D, D-Special, CD, CHD etc.)}
The histogram shows a unimodal and highly leptokurtic structure, thus we attempt to fit a single log-normal distribution. However, it fits badly. A mixture of two log-normal distributions fits reasonably better, yielding a p-value of 0.4998 for Pearsonian goodness of fit test, and a p-value of 0.324 for Kolmogorov Smirnov test. Hence we report a \textbf{mixture of two log-normal distributions} with the following properties. The closeness of the means and smaller variances in both the components explain the reason for a unimodal looking structure in the histogram.

\begin{center}
    \begin{tabular}{|c|c|c|c|c|}
    \hline 
    Proportion & $\mu$ & $\sigma$ & Expectation & Variance\\
    \hline 
    0.8944309 & 0.08910557 & 0.05939173 & 1.095126 & 0.004237857 \\
    0.1055691 & 0.08829919 & 0.11626259 & 1.099722 & 0.016458284 \\
    \hline 
    \end{tabular}
\end{center}

\subsection{For Cluster 5 (PD-Fine)}
In this case, as the histogram shows a bimodal shape, we try fitting with a 2 component mixture of log-normal distribution. It gives reasonably well p-value, 0.8231 for Pearson's chi-sqaured goodness of fit test and 0.9891 for Kolmogorov Smirnov's test. Hence we report the distribution of the ratio to be a \textbf{mixture of two log-normal distributions} with the following parameters. 
\begin{center}
    \begin{tabular}{|c|c|c|c|c|}
    \hline
    Proportion & $\mu$ & $\sigma$ & Expectation & Variance\\
    \hline 
    0.4000716 & -0.004682854 & 0.03456505 & 0.9959229 & 0.001185728 \\
    0.5999284 & 0.125364986 & 0.05488756 & 1.1352709 & 0.00388671\\
    \hline 
    \end{tabular}
\end{center}

\subsection{For Cluster 6 (OPD, OPD-Clonal, ORD, D1 etc.)}
The histogram for this cluster is characterized by its unimodality and slight positive skewness. A single log-normal yielded a p-value of 0.0104. Hence, we moved on to a mixture of two log-normal distributions, and this time, the p-value came out to be a higher value of 0.06225. We also tried to fit a three-component mixture of log-normal distributions, however, that does not increase the p-value by a significant amount. Thus we report a \textbf{mixture of two log-normal distributions} with the following properties, to maintain the simplicity of the underlying model.

\begin{center}
    \begin{tabular}{|c|c|c|c|c|}
    \hline
    Proportion & $\mu$ & $\sigma$ & Expectation & Variance\\
    \hline
    0.09208835 & -0.00104782 & 0.03631677 & 0.9996117 &  0.001318753\\
    0.90791165 &  0.08659977 & 0.10302695 & 1.0962629 & 0.012824430\\
    \hline
    \end{tabular}
\end{center}

\section{Analysis of the Price Model}

To model the pricing system of the tea market, we consider modeling the demand side by consideration of the Valuation of tea grades, its grade, source, and the month in which the tea lots are available. On the other hand, to model the supply side of the market, we consider the volume of the tea lots as our main predictor. Therefore, our pricing model should include these variables.

To check whether the variant of the tea gardens (Clonal, Gold, Royal, Special, etc.) should be included in the pricing model, we simply fit a one-way Analysis of Variance model with Price as our response variable and the variant of tea garden as the possible treatment variable. It is found that this factor explains a sum of squares of 842405 with 4 degrees of freedom, yielding an F-statistic value of 102.41 and consequently extremely small p-value. Therefore, based on the data, we find sufficient evidence to incorporate this factor into our pricing model in order to have better predictability.

\subsection{Nonparametric Exploration of limitations of a Pricing Model}

Before specifying a statistical model for the prediction of Valuation and Price of a tea packet based on several of its characteristics, it is extremely important to understand the limitations of how much we can do first. To this end, it is found that there may be several tea packets with exactly the same characteristics, (coming from the same garden, is of the same grade and comes in the same week), and between them, the Range of their Valuations and Prices are calculated. If we consider the empirical CDF of such all possible ranges, then as seen from Figure \ref{fig:ecdf_range}, in most cases, those packets are subjected to exactly the same Valuation by the auctioneers, however, the final Prices at they are sold may be very different. This exploration suggests that if we simply do away with Valuation and output a single prediction of Price for tea packets based on its Grade, Source, and Week of the year when it is held for the auction, we can almost hope for $55\%$ accuracy in prediction. However, if we wish to have a $90\%$ accuracy in predicting price, we must allow approximately $30$ rupees of deviation from the actual price. Even with a more robust measure of dispersion, mean deviation about median, the conclusion of this exploration remains the same, as seen from Figure \ref{fig:ecdf_md}.

\begin{figure}[ht]
    \centering
    \includegraphics[width = \textwidth]{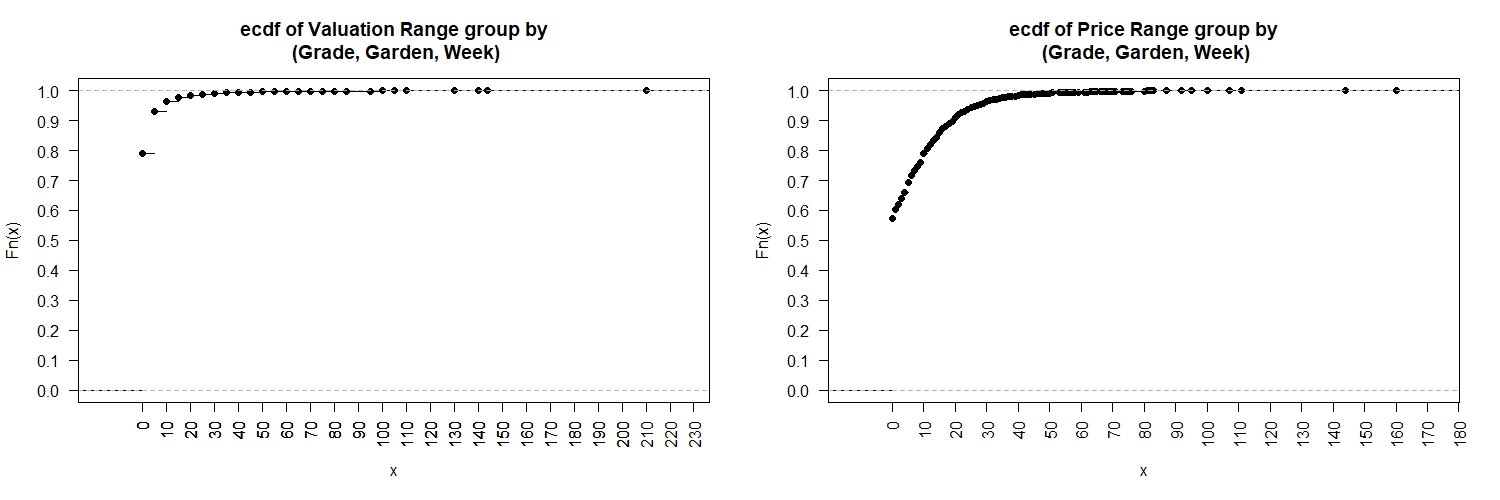}
    \caption{[Left] Empirical cumulative distribution function of the Range of Valuations of tea packets with exactly same characteristics; [Right] ecdf of the Range of Prices of tea packets with exactly same characteristics;}
    \label{fig:ecdf_range}
\end{figure}

\begin{figure}[ht]
    \centering
    \includegraphics[width = \textwidth]{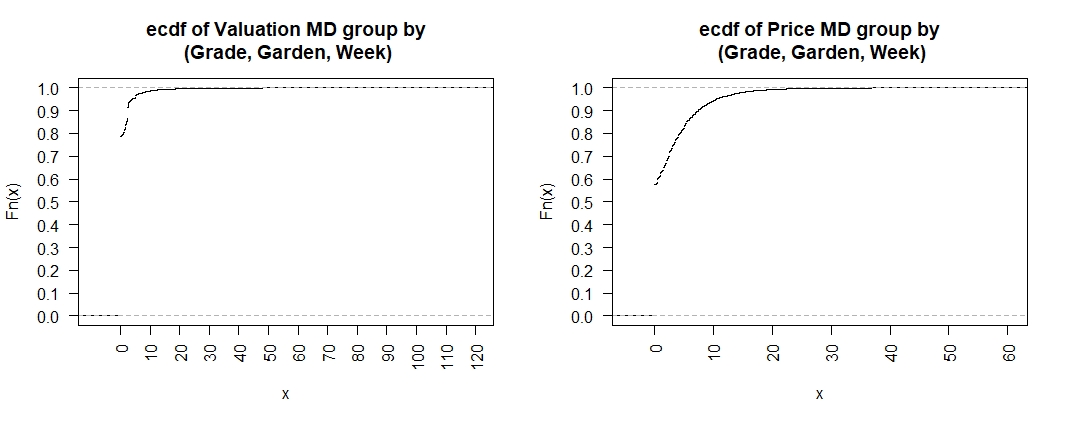}
    \caption{[Left] ecdf of the Mean deviation of Valuations of tea packets with exactly same characteristics; [Right] ecdf of the Mean deviation of Prices of tea packets with exactly same characteristics;}
    \label{fig:ecdf_md}
\end{figure}

Since there are lots of Grades, Source and possible weeks combination provided in the dataset, modeling a different prediction for each of these combinations would require a great number of parameters in the model. However, as we have already performed a reasonable clustering analysis of different tea grades and the source gardens, we wish to explore whether having a prediction for tea grades with similar cluster characteristics will be within the economical tolerance level for the auctioneers. Figure \ref{fig:ecdf_range_cluster} and \ref{fig:ecdf_md_cluster} shows the corresponding ECDFs of Ranges and Mean deviations about median of Valuations and Prices, for tea packets sharing same cluster characteristics. As seen from them, such a model that predicts at the cluster level would be inadequate in modeling either the valuation or the price soundly.

\begin{figure}[ht]
    \centering
    \includegraphics[width = \textwidth]{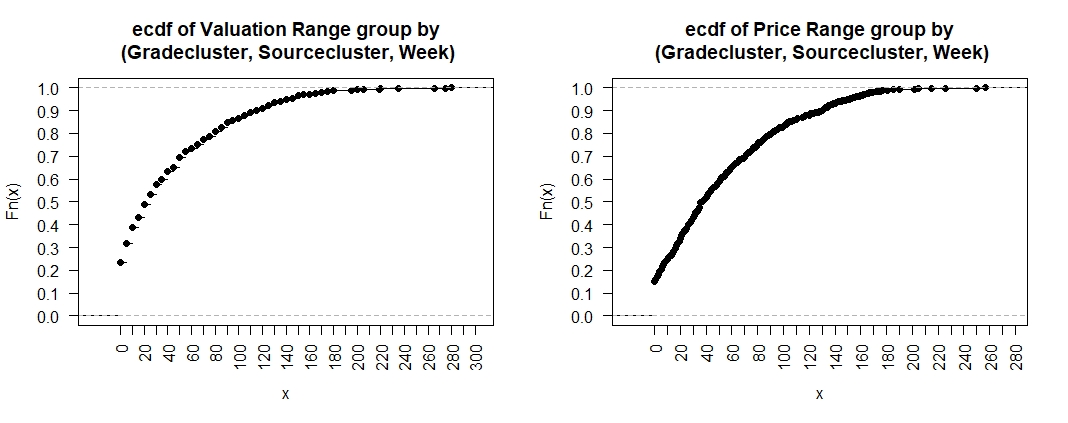}
    \caption{[Left] Empirical cumulative distribution function of the Range of Valuations of tea packets with exactly same cluster characteristics; [Right] ecdf of the Range of Prices of tea packets with exactly same cluster characteristics;}
    \label{fig:ecdf_range_cluster}
\end{figure}

\begin{figure}[ht]
    \centering
    \includegraphics[width = \textwidth]{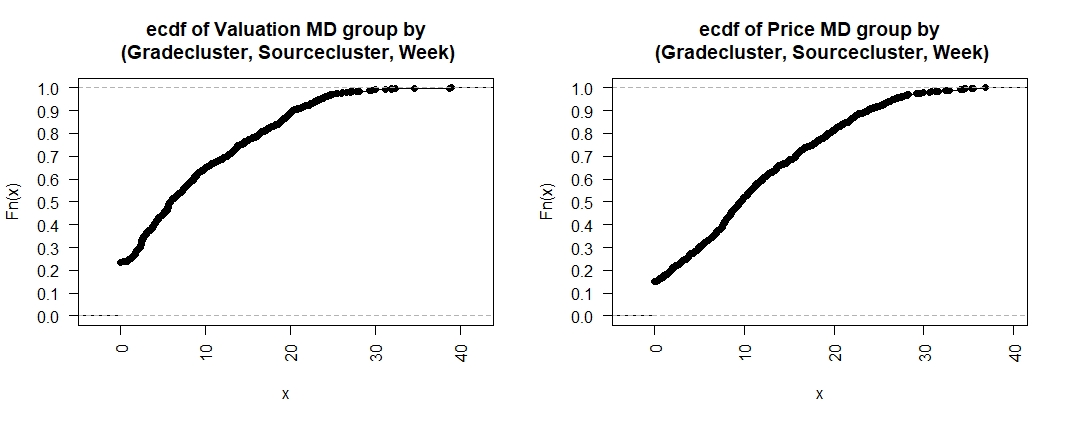}
    \caption{[Left] ecdf of the Mean deviation of Valuations of tea packets with exactly same cluster characteristics; [Right] ecdf of the Mean deviation of Prices of tea packets with exactly same cluster characteristics;}
    \label{fig:ecdf_md_cluster}
\end{figure}

\subsection{Linear Price Model}

For ease of interpretability, we start with a simple linear model for pricing system with the grade, source clusters, month of availability,  variant of source garden, the volume of the tea packet, and valuation of the tea packets as our predictors. The model is given by; 
\begin{align*}
    \Omega_1&: \text{Price}=\beta_0+\underbrace{\beta_\text{Grade}+\beta_\text{Source}+\beta_\text{Month}+\beta_\text{Garden}+\beta_1\text{Valuation}}_{\text{Factors for Demand}} +\underbrace{\beta_2\text{Volume}}_{\text{Supply}}+\varepsilon\\
    & \text{ where } \varepsilon\sim N(0,\sigma^2) \text{ independently and identically distributed}
\end{align*}

The results from this model are summarized in Table ~\ref{linear-pricemodel-val}.

\begin{table}
\centering
\caption{Analysis of Variance table for fitted linear pricing model with Valuation}
\label{linear-pricemodel-val}
\begin{tabular}{|c|c|c|c|c|c|}
\hline
               & \begin{tabular}[c]{@{}c@{}}Degrees of \\ Freedom\end{tabular} & Sum Sq.  & Mean Sq. & F Value  & p value \\ \hline
Grade  & 5  & 20140487 & 4028097  & 25153.19 & 2.2e-16 \\ \hline
Source  & 6 & 1696940  & 282823   & 1766.07  & 2.2e-16 \\ \hline
Month  & 10 & 5473468  & 547347   & 3417.87  & 2.2e-16 \\ \hline
Volume & 1 & 62247    & 62247    & 388.7   & 2.2e-16 \\ \hline
Source Variant & 4 & 186240   & 46560  & 290.74 & 2.2e-16 \\ \hline
Valuation & 1 & 13191583  & 13191583  & 82373.97 & 2.2e-16 \\ \hline
Residuals & 21106 & 3379970 & 160 &          &         \\ \hline
\end{tabular}
\end{table}

\begin{figure}
    \centering
    \includegraphics[width=\textwidth]{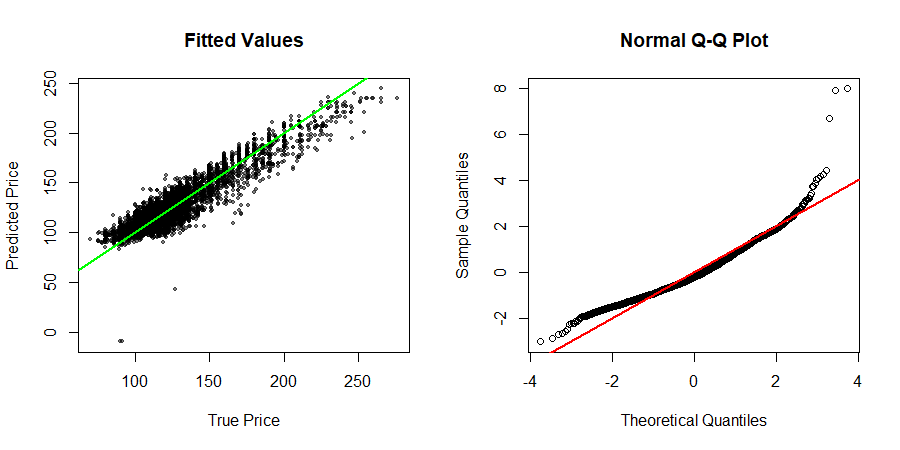}
    \caption{[Left] True Price vs Predicted Price for the fitted linear pricing model; [Right] The qqplot of standardized residuals for residual diagnostics}
    \label{linear-pricemodel-val-fig}
\end{figure}

\begin{enumerate}
    \item Multiple R-squared for the fitted model is 0.9234, suggesting a very strong linear relationship among the predictors.
    \item The Analysis of Variance decomposition between different effects has been shown in Table ~\ref{linear-pricemodel-val}. Clearly, each of the predictor variables explains a lot of variation in the price, and the extremely small p-values indicate that each one of the variables' contribution is significant.
    \item The true price and predicted price based on the fitted linear model has been shown in Figure ~\ref{linear-pricemodel-val-fig}. Note that, the predicted prices are randomly dispersed on both sides of the reference line, thereby supporting the assumption of homoscedasticity. Also, the standardized residuals have quantiles similar to that of the theoretical quantiles of a standard normal distribution, as assumed by the model, other than some drastic outliers present in both ends.
    \item Based on the evaluation of the pricing model in the testing dataset, the 2.5\% quantile of the residuals is -18.59145, while the 97.5\% quantile of the residuals is 23.43802. Hence, this pricing model approximately makes an error about 20 Rupees, to both positive and negative sides, considering a robust measure of variation.
    \item On the other hand, considering a classical approach to measure the standard error, we find the interval, with predicted value as the center and an error of 24.3229 added (or subtracted) to both sides contains the true price for 95\% of the time.
\end{enumerate}

Comparatively, proceeding with our main objective, we remove the valuation as predictor to see how much it affects our original linear pricing model. In this case, we find that a simple linear model with Price of the tea packets as response variable would make the residuals to be heteroscedastic. Therefore, we apply a variance stabilizing logarithmic transformation and use the natural logarithm of price of tea packets as response variable. Thus the model is given by; 
\begin{align*}
    \Omega_2&: \log\left(\text{Price}\right)=\beta_0+\underbrace{\beta_\text{Grade}+\beta_\text{Source }+\beta_\text{Month}+\beta_\text{Garden}}_{\text{Factors for Demand}}+\underbrace{\beta_2\text{Volume}}_{\text{Supply}}+\varepsilon
\end{align*}
where $\varepsilon\sim N(0,\sigma^2)$ independently and identically distributed.
The results obtained are summarized in Table ~\ref{linear-pricemodel-noval}

\begin{table}
\centering
\caption{Analysis of Variance table for fitted linear logarithmic pricing model without Valuation}
\label{linear-pricemodel-noval}
\begin{tabular}{|c|c|c|c|c|c|}
\hline
               & \begin{tabular}[c]{@{}c@{}}Degrees of \\ Freedom\end{tabular} & \begin{tabular}[c]{@{}c@{}}Sum Sq.\\ (log scale)\end{tabular} & \begin{tabular}[c]{@{}c@{}}Mean Sq.\\ (log scale)\end{tabular} & F Value & p value \\ \hline
Grade  & 5                                                             & 835.15                                                        & 167.030                                                        & 5566.005 & 2.2e-16 \\ \hline
Source         & 6                                                             & 52.15                                                         & 8.684                                                         & 289.375 & 2.2e-16 \\ \hline
Month          & 10                                                           & 248.20                                                        & 24.820                                                        & 827.086 & 2.2e-16 \\ \hline
Volume         & 1                                                            & 0.564                                                           & 0.564                                                         & 18.779  & 1.475e-5 \\ \hline
Source Variant & 4                                                        & 5.37                                                          & 1.342                                                          & 44.728  & 2.2e-16 \\ \hline
Residuals      & 21107                                                         & 633.40                                                        & 0.030                                                        &         &         \\ \hline
\end{tabular}
\end{table}

\begin{figure}
    \centering
    \includegraphics[width=\textwidth]{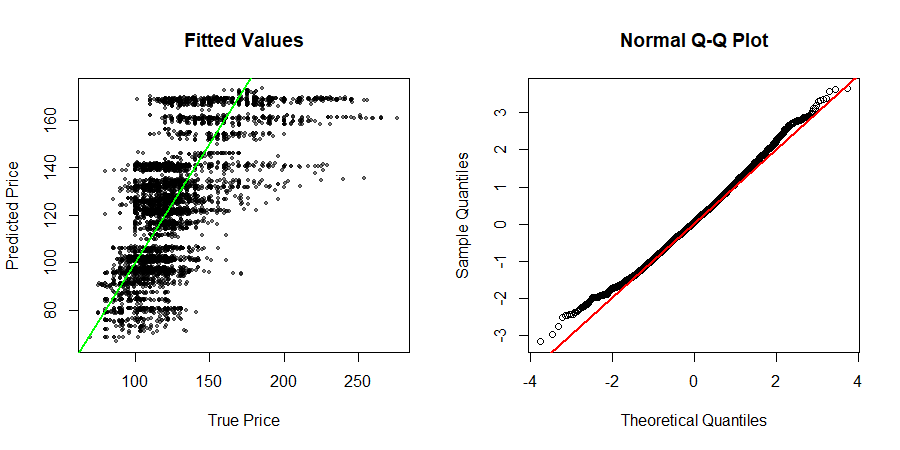}
    \caption{[Left] True Price vs Predicted Price for the fitted log-linear pricing model; [Right] The qqplot of standardized residuals for residual diagnostics}
    \label{linear-pricemodel-noval-fig}
\end{figure}

\begin{enumerate}
    \item Multiple R-squared for the fitted model is 0.6431, suggesting a moderately strong linear relationship among the predictors.
    \item The Analysis of Variance decomposition between different effects has been shown in Table ~\ref{linear-pricemodel-noval}. Clearly, each of the predictor variables explains a lot of variation in the price, and the extremely small p-values indicate that each one of the variables' contribution is significant. However, to our surprise, the effect of Volume has been explained by other variables to some extent, although the volume of tea packets generates the supply side of the market.
    \item The true price and predicted price based on the fitted log-linear model has been shown in Figure ~\ref{linear-pricemodel-noval-fig}. Note that, the amount of error in prediction increases as the true price increases, thereby supporting the evidence of heteroscedasticity as previously mentioned. Also, from the Q-Q plot, it is evident that there are some possible outliers present at lower prices.
    \item Testing the performance of the fitted model on the testing set, the 2.5\% quantile of the residuals comes as -0.29539, while the 97.5\% quantile of the residuals is 0.37754 in logarithmic scale. Hence, this pricing model approximately makes an error from 74.43\% to 145.86\% of the true prices of the tea packets. This is based on a robust approach to estimate the standard error.
    \item On the other hand, considering a classical approach to measure the standard error, we find that the predicted prices lie between 71.09\% and 140.65\% of the true prices about 95\% of the time.
\end{enumerate}

Thus the valuation part is significant for the prediction of the prices, and the process cannot be automated by a simple linear or log-linear model approach.

\section{Causal Analysis between Price and Valuation}

From the simple statistical analysis with a linear model performed above, it seems Valuation is indeed a pertinent variable in the explanation of the Pricing system. However, we could not yet allege that the auctioneers' valuation had a causal impact on the price level, although it seemed to be an indispensable predictor. To evaluate the causal impact of valuation on price, we lay down a part of the causal graph structure in Figure \ref{flowchart:tea}. Note that there may be some causation structure between the variables (for example, not all grades occur in all months, hence there should be an arrow from month to grade). Nevertheless, the variables of interest are the valuation and the final price, and the parents to them are known, assuming we subscribe to a linear understanding of time and causality. Thus all the other variables in the model are parents to both valuation and price. Furthermore, valuation could be a parent of price. But, if the valuation is indeed an 'educated guess' of price as was originally intended, then, since the guess is based on only these variables; conditioned on these variables, price and valuation should be independent trials from (possibly same) distribution. Hence, they should be independent, had valuation not been a cause of price. Hence we aim to test this independence.\par

\subsection{Causal Analysis with linear effect model}

\begin{figure}[ht]
    \centering
    \includegraphics[width = 0.8\textwidth]{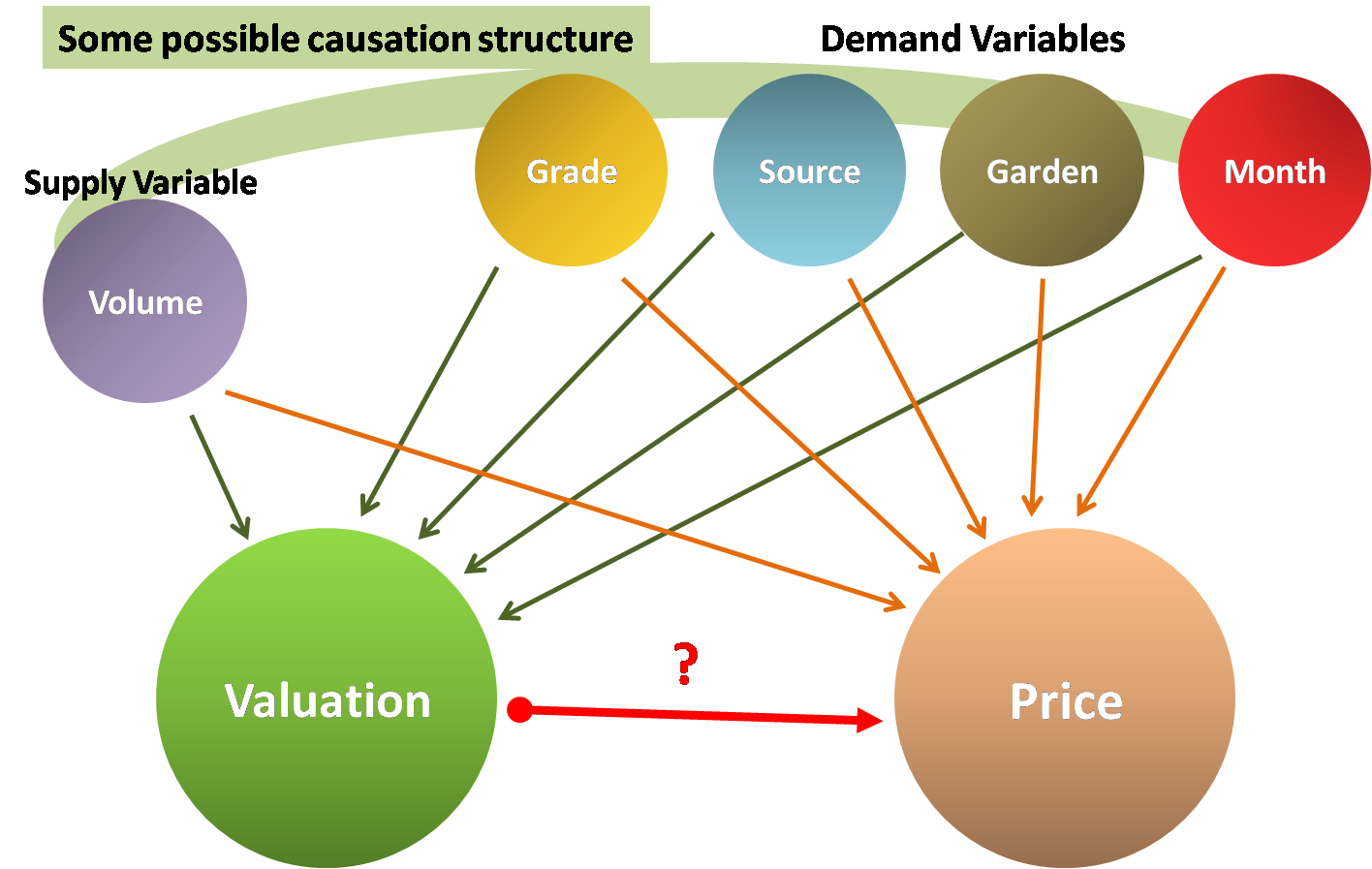}
    \caption{Causal Diagram for modeling Tea auction process}
    \label{flowchart:tea}
\end{figure}

Thus, we wish to test whether $$\text{Valuation}\indep \text{Price} \mid (\underbrace{\text{ Source, Grade, Volume, Garden, Month}}_{\text{rest}}) ?$$ However, we are faced with the fact that conditioning on so many variables render fewer data to provide reliable estimates for any inference. Hence, we take a different strategy, which is generally often taken when the conditioning variables are continuous. We model the log of valuations by a linear function of the other variables, and similarly for the log of price. That is 

\begin{align*}
    \Omega_3: \log \text{Valuation} & = \beta_1^\top \text{rest} + \varepsilon_1, \hspace{1cm}\varepsilon_1\sim \mathcal{N}(0,\sigma_1^2) \\
    \log \text{Price} & = \beta_2^\top \text{rest} + \varepsilon_2\hspace{1cm}\varepsilon_2\sim \mathcal{N}(0,\sigma_2^2), \ \varepsilon_1\indep \varepsilon_2
\end{align*}

If these models provide a good fit, then we can look at the correlation between the residuals to identify the presence of a causal link between valuation and price. This is because the residuals in both the models are independent of the rest, and for two variables $A, B$ and $C$, we have, 
$$A\indep B |C \Leftrightarrow [A - f(C)]\indep [B - g(C)] | C \Leftrightarrow A - f(C) \indep B -  g(C) $$ if $A-f(C)$ and $B-g(C)$ are independent of $C$ itself. This is akin to testing for the partial correlation between valuation and price to be 0. \par 

The results that we obtain from the data are as follows: 
\begin{itemize}
    \item $\Omega_3$ provides a reasonably good fit, yielding the value of the multiple $R^2$ to be $0.75$ and $0.64$ respectively for the two equations. The regression diagnostics checks were satisfied.
    \item The correlation between the residuals of the two regression equations of $\Omega_3$ came out to be a whopping 0.859; which cannot be attributed to sheer chance with such a large sample size (a sample size of about $19,000$).
\end{itemize}

This has widespread implications. First of all, this substantiates that valuation of the auctioneers is not as harmless as providing an educated guess for the price, but rather have a causal impact on the final price. This occurs, as the base price is visible to all potential buyers. Thus, to automate the process, whatever procedure we propose cannot be held against the standard with how the valuation predicts price, as the predictability is nested as a causal impact. Hence, to automate the process, a possible change in the auction mechanism is required, and the premises of checking the success of any such alternate mechanism with this data (where valuation has had a causal impact) would be flawed.

\subsection{Causal Analysis with Three-Stage Latent Hierarchical Causal Model}

The linear analysis provides substantial evidence in favor of the price- valuation linkage. But this may be confounded with many other interactions present in the model. Recall that we have $6$ clusters for types of tea grades and $7$ clusters for source or tea gardens from where the packet came from. Together, considering their interactions, we have a $6\times 7 = 42$ element matrix. These cells are the hidden states in the model and can be interpreted as the proxies for demands for different types of tea in the market. These states, along with several other predictors will generate a common value about a single grade, garden, and week combination, which shall be the true value of the tea packet to both the auctioneers and the buyers. Finally, these prices and valuations will be characterized based on the common value and shall further be influenced by the particular volume of the lot, which would yield the final observations. The model that describes such a situation most closely is a variant of the Linear dynamical system model, popularly associated with Kalman filter, as described in \cite{kalman1960new} and \cite{kalman1961new}.

The mathematical framework is as follows;

\begin{equation}
    Z_t = FZ_{t-1} + \xi_t \qquad \xi_t \sim N(0, Q) \quad \forall t = 1, 2, \dots T
    \label{eqn:model-Zt}
\end{equation}

where $Z_t$ is a $42\times 1$ vector, denoting the market condition. Here, no intercept term is used, since we wish the matrix $F$ to be interpreted as a transition matrix over the market condition vector $Z_t$. One can simply identify $Z_t$ to be a non-deterministic linear dynamic system. Let, $Z_t$ be denoted symbolically as;

$$
Z_t = 
\begin{bmatrix}
a_{11, t}\\
a_{12, t}\\
\vdots\\
a_{17, t}\\
a_{21, t}\\
\vdots\\
a_{67, t}
\end{bmatrix}
$$

where $a_{ij, t}$ is the latent market condition for the demand of tea in the $i$-th grade cluster and $j$-th source cluster, at time $t$. The next level of the model is;

\begin{multline}
    W_{gt} = \Phi_0 + \Phi W_{g(t-1)} + G_g Z_t + H X_{gt} + e_{gt} \qquad e_{gt} \sim N(0, R)\\
    \forall g = 1, 2, \dots N_t; \quad \forall t = 1, 2, \dots T
    \label{eqn:model-Wgt}
\end{multline}

where $W_{gt}$ is a scalar, which denotes the common value of the tea lot at the combination $g$ = (Garden, Grade) at time $t$, which depends on its previous observation, the current market state $Z_t$ and some exogenous control variables $X_{gt}$. In this case, the vector $G_g$ has a special structure such that;

$$G_g Z_t = \beta_1 a_{i_0, j_0, t} + \beta_2 \sum_{i \neq i_0} a_{i, j_0, t} + \beta_3 \sum_{j \neq j_0} a_{i_0, j, t}$$

where $\beta_1, \beta_2$ and $\beta_3$ are parameters to be estimated. Here, $g$ is a grade and garden combination such that the grade belongs to the $i_0$-th cluster, and the garden belongs to the $j_0$-th cluster of the source. This special structure means that the common value for a tea packet depends on the market condition of demand for that particular type of tea, as well as the market condition of its available substitutes, which shares either the same tea grade or the same source garden, as a potential cause for substitutability.

And finally, we have the model for the observations;

\begin{multline}
    y_{igt} = W_{gt}\mathbf{1}_2 + \Gamma u_{igt} + \epsilon_{igt} \qquad \epsilon_{igt} \sim N(0, S)\\
    \forall i = 1, 2, \dots r_{gt}; \quad \forall g = 1, 2, \dots N_t; \quad \forall t = 1, 2, \dots T
    \label{eqn:model-yigt}
\end{multline}

where $y_{igt}$ is the actual bivariate observation of Price and Valuation of $i$-th repeated measure in $g$-th group combination in $t$-th time, while $u_{igt}$ is some more exogenous variables, whose influences are incorporated only in the final stage and;

$$\mathbf{1}_2 = \begin{bmatrix}
1\\
1\\
\end{bmatrix}$$

The reason we need to deviate from the standard Simple Adaptive Control Model \cite{meyntweedie} (Page 40) or the popularly known Kalman Filter model (which allows for only two indices), is that we have more than one observations which are manifestations of the same state (i.e., three indices), and hence, a direct influence of the states do not account for the variability within the observations of the same state.

In the above specification of our three-stage model, we only observe the variables $y_{igt}, u_{igt}$ and $X_{gt}$, and the latent variables $Z_t, W_{gt}$ are unobservable. Hence, we can characterize the model with the specification of all the parameters and the unobservable latent variables, namely by the list of elements $(Z_t, W_{gt}, F, Q, \Phi_0, \Phi, G_g, H, R, \Gamma, S)$. Unfortunately, the above model is not identifiable, as the new set of elements given by $(\alpha Z_t, W_{gt}, F, \alpha^2 Q, \Phi_0, \Phi, \frac{1}{\alpha} G_g, H, R, \Gamma, S)$, also result in the exact same model. The crucial reason for this unidentifiablity is that the first equation contains no observable variable. For this reason, we require to pose a constraint on the model by specifying $\Vert Z_t \Vert = 1$, i.e. the vector $Z_t$'s are normalized for any $t = 0, 1, \dots T$. 

Figure \ref{fig:causal-dag-3-stage} shows the Causal DAG diagram for the above three-stage latent hierarchical (TSLH) model, for a fixed time point $t$, the defining SCM for this are given by equations \ref{eqn:model-Zt}, \ref{eqn:model-Wgt} and \ref{eqn:model-yigt}. Note that, there are $4$ exogenous variables at the second stage (as noted from the four parameters $H_1, H_2, H_3$ and $H_4$) and only one exogenous variable in the third stage. We use Gibb's sampler to obtain the estimates, as discussed in the following subsection. The description of these parameters, along with the estimated value from the dataset is given in table \ref{tbl:3-stage-parameters}.

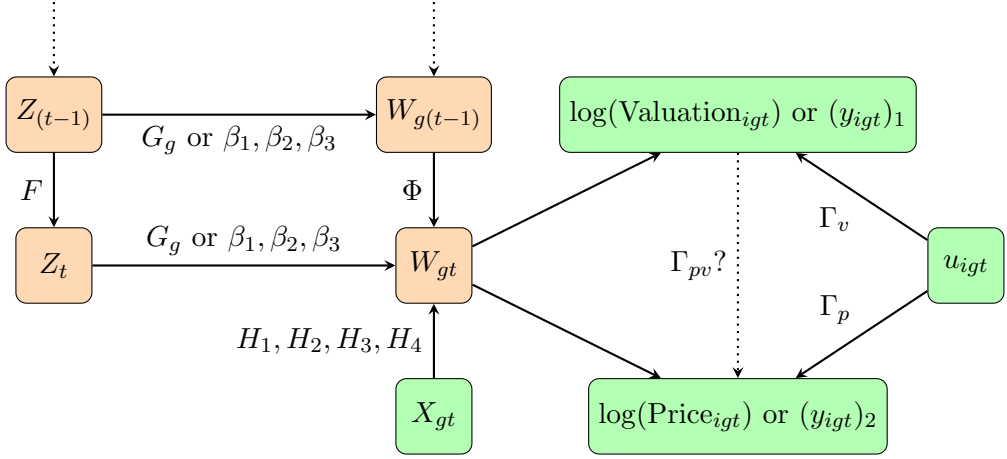
\begin{figure}[ht]
    \centering
    \begin{tikzpicture}
        \node (Z_cur) [latents] {$Z_t$};
        \node (Z_prev) [latents, above of=Z_cur, yshift = 1cm] {$Z_{(t-1)}$};
        \node (W_cur) [latents, right of=Z_cur, xshift = 4cm] {$W_{gt}$};
        \node (W_prev) [latents, right of=Z_prev, xshift = 4cm] {$W_{g(t-1)}$};
        \node (Z_past) [invisible, above of=Z_prev, yshift = 1cm] {};
        \node (W_past) [invisible, above of=W_prev, yshift = 1cm] {};
        \node (X_cur) [obsvars, below of=W_cur, yshift = -1cm] {$X_{gt}$};
        \node (Val_cur) [obsvars, right of=W_cur, yshift = 2cm, xshift = 3cm ] {$\log(\text{Valuation}_{igt})$ or $(y_{igt})_1$};
        \node (Price_cur) [obsvars, right of=W_cur, yshift = -2cm, xshift = 3cm ] {$\log(\text{Price}_{igt})$ or $(y_{igt})_2$};
        \node (u_cur) [obsvars, right of=W_cur, xshift = 6cm ] {$u_{igt}$};

        \draw [arrow] (Z_prev) -- node[anchor=east] {$F$} (Z_cur);
        \draw [arrow, dotted] (Z_past) -- (Z_prev);
        \draw [arrow, dotted] (W_past) -- (W_prev);
        \draw [arrow] (Z_prev) -- node[anchor=north] {$G_g$ or $\beta_1, \beta_2, \beta_3$} (W_prev);
        \draw [arrow] (Z_cur) -- node[anchor=south] {$G_g$ or $\beta_1, \beta_2, \beta_3$} (W_cur);
        \draw [arrow] (W_prev) -- node[anchor=east] {$\Phi$} (W_cur);
        \draw [arrow] (X_cur) -- node[anchor=east] {$H_1, H_2, H_3, H_4$} (W_cur);
        \draw [arrow] (u_cur) -- node[anchor=north east] {$\Gamma_{v}$} (Val_cur);
        \draw [arrow] (u_cur) -- node[anchor=south east] {$\Gamma_{p}$} (Price_cur);
        \draw [arrow] (W_cur) -- (Val_cur);
        \draw [arrow] (W_cur) -- (Price_cur);
        \draw [arrow, dotted] (Val_cur) -- node[anchor=east] {$\Gamma_{pv}$?} (Price_cur);

    \end{tikzpicture}

    \caption{Causal DAG for the three stage latent hierarchical model }
    \label{fig:causal-dag-3-stage}
\end{figure}

\subsubsection{Gibbs Sampling Conditional Derivations}

We begin by writing the likelihood for the Three Stage Latent Hierarchical (TSLH) model, upto a proportionality constant.

\begin{align*}
    \mathcal{L} 
    & \propto \prod_{i, g, t} {\vert S\vert}^{-1/2} \exp\left[-\frac{1}{2} (y_{igt} - W_{gt}\mathbf{1}_2 - \Gamma u_{igt} )^{\top}S^{-1} (y_{igt} - W_{gt}\mathbf{1}_2 - \Gamma u_{igt} )\right] \times \\
    & \qquad \prod_{g, t} {\vert R\vert}^{-1/2} \exp\left[-\frac{1}{2} (W_{gt} - \Phi_0 - \Phi W_{g(t-1)} - G_g Z_t - HX_{gt} )^{\top}\right. \\ 
    & \qquad \qquad \qquad \qquad \qquad \bigg. R^{-1} (W_{gt} - \Phi_0 -\Phi W_{g(t-1)} - G_g Z_t - HX_{gt} )  \bigg] \times \\
    & \qquad \prod_t {\vert Q\vert}^{-1/2} \exp\left[-\frac{1}{2} (Z_{t} - FZ_{(t-1)} )^{\top}Q^{-1} (Z_{t} - F Z_{(t-1)} )  \right]
\end{align*}

Before obtaining the individual conditional distributions, the following observation will come in handy: 

Note that, when $Y\sim \mathcal{N}_k(\mu,\Sigma)$ as in regression model, then 
$$f(\mathbf{y})\propto \exp\left(-\dfrac 12Y^\top \Sigma^{-1} Y + Y^\top \Sigma^{-1}\mu \right)$$ 
Therefore, we can simply identify the normal distribution based on these coefficients $\Sigma^{-1}$ and $\Sigma^{-1}\mu$, which is a reparametrization of the parameters of normal distribution. This reparametrization shall be helpful in identifying the conditional distributions in the subsequent calculations. 

We first obtain the conditional distributions for the latent variables, $Z_t$ and $W_{gt}$ respectively.

\begin{multline*}
    Z_t\mid \text{rest} \propto \exp\left[ -\dfrac{1}{2} \left\{  Z_t^{\top}\left(Q^{-1} + F^{\top}Q^{-1}F + \sum_g G_g^{\top}R^{-1}G_g  \right) Z_t \right.\right.\\
    \left.\left.
    - 2Z_t^{\top} \left( Q^{-1}FZ_{t-1} + F^{\top}Q^{-1}Z_{t+1} + \sum_g G_g^{\top} R^{-1} \left( W_{gt} - \Phi_0 - \Phi W_{g(t-1)} - HX_{gt} \right) \right) 
    \right\} \right]
\end{multline*}

which is a normal distribution with the above reparametrization where, $\Sigma^{-1}$ is the coefficient of quadratic term, and $\Sigma^{-1}\mu$ is the coefficient of the linear term.
 
Now, with $W_{gt}$, we have;

\begin{multline*}
    W_{gt} \mid \text{rest} \propto
    \exp\left[-\dfrac{1}{2}\left\{
        W_{gt}^{2} \left( \dfrac{1}{R} + \dfrac{\Phi^2}{R} + r_{gt}\mathbf{1}_2^{\top}S^{-1}\mathbf{1}_2 \right) 
        - 2 W_{gt} \left( \dfrac{\Phi_0 +\Phi W_{g(t-1)} + G_g Z_t }{R} \right.\right.\right.\\
        \left.\left.\left.
        + \dfrac{ HX_{gt}+\Phi(W_{g(t+1)} - \Phi_0 - G_g Z_{t-1} - HX_{g(t-1)} )}{R} + \mathbf{1}_2^{\top}S^{-1}(y_{igt} - \Gamma u_{igt}) \right)
    \right\}\right]
\end{multline*}

This leads to another normal distribution. Next, for the parameters, 

$$Q \mid \text{rest} \propto \vert Q\vert^{-T/2} \exp\left[ -\dfrac{1}{2} \text{tr}\left( \left(\sum_t \xi_t \xi_t^{\top} \right) Q^{-1} \right) \right]$$

which leads to $\mathcal{W}^{-1}\left[ \sum_t \xi_t \xi_t^{\top}; T-43 \right]$ distribution, where $\mathcal{W}^{-1}$ is used to denoted Inverse Wishart distribution.

Similarly, 

$$R \mid \text{rest} \propto R^{-\sum_{t}N_t /2} \exp\left[ -\dfrac{1}{2}  \left(\sum_{g,t} e_{gt} e_{gt}^{\top} \right) R^{-1} \right]$$

which is same as the density function of Inverse Gamma distribution with shape $\alpha = \sum_t \dfrac{N_t}{2} - 1$, scale parameter $\beta = \dfrac{1}{2} \left(\sum_{g,t} e_{gt} e_{gt}^{\top} \right)$, upto a proportionality constant.

And finally, 

$$S \mid \text{rest} \propto \vert S\vert^{-N/2} \exp\left[ -\dfrac{1}{2} \text{tr}\left( \left(\sum_{i, g, t} \epsilon_{igt} \epsilon_{igt}^{\top} \right) S^{-1} \right) \right]$$

which again leads to $\mathcal{W}^{-1}\left[ \sum_{i, g, t} \epsilon_{igt} \epsilon_{igt}^{\top}; N-3 \right]$ distribution.

Continuing,

\begin{align*}
    F \mid \text{rest} & \propto \exp\left[ -\dfrac{1}{2} \left\{ 
    \text{tr}\left( \sum_t Z_{t-1}^{\top}F^{\top}Q^{-1}FZ_{t-1} \right) - 2\text{tr}\left( \sum_t Z_{t-1}^{\top}F^{\top}Q^{-1}Z_t \right) \right\} \right]\\
    & \propto \exp\left[ -\dfrac{1}{2} \left\{ 
        \text{tr}\left( Q^{-1} \left(\sum_t Z_{t-1}Z_{t-1}^{\top}\right) F^{\top}F \right) - 2\text{tr}\left( Q^{-1} F^{\top} \left(\sum_t Z_t Z_{t-1}^{\top}\right) \right) \right\} \right]
\end{align*}

therefore,

$$F \mid \text{rest} \sim \mathcal{MN}_{42\times 42}\left( \left[ \sum_t Z_{t-1}Z_{t-1}^{\top}\right]^{-1} \left[ \sum_t Z_t Z_{t-1}^{\top}\right], I, \left[ \sum_t Z_{t-1}Z_{t-1}^{\top}\right]^{-1} Q \right)$$

where $\mathcal{MN}$ stands for the matrix normal distribution. In other words, since the covariance matrix between the rows of the $F$ is $I$, the identity matrix, hence we can generate the rows of $F$ independently from multivariate normal distributions with mean vectors same as the rows of the mean matrix, and the same covariance matrix $\left[ \sum_t Z_{t-1}Z_{t-1}^{\top}\right]^{-1} Q$.

On a similar note, 

\begin{align*}
    \Gamma \mid \text{rest} & \propto \exp\left[ -\dfrac{1}{2} \left\{ 
        \text{tr}\left( \sum_{i, g, t} u_{igt}^{\top}\Gamma^{\top}S^{-1}\Gamma u_{igt} \right) - 2 \text{tr}\left( \sum_{i, g, t} u_{igt}^{\top}\Gamma^{\top}S^{-1}\left( y_{igt} - W_{gt}\mathbf{1}_2 \right) \right) 
    \right\} \right]\\
    & \propto \exp\left[ -\dfrac{1}{2} \left\{ 
        \text{tr}\left( S^{-1} \left(\sum_{i, g, t} u_{igt} u_{igt}^{\top}\right)\Gamma^{\top}\Gamma \right) - 2 \text{tr}\left( S^{-1} \Gamma^{\top} \sum_{i, g, t} \left( y_{igt} - W_{gt}\mathbf{1}_2 \right)u_{igt}^{\top} \right) 
    \right\} \right]\\
\end{align*}

Therefore,

$$\Gamma \mid \text{rest} \sim \mathcal{MN}_{2\times 2}\left(  \left[\sum_{i, g, t} u_{igt} u_{igt}^{\top}\right]^{-1} 
\left[\sum_{i, g, t} \left( y_{igt} - W_{gt}\mathbf{1}_2 \right)u_{igt}^{\top}\right], I, \left[\sum_{i, g, t} u_{igt} u_{igt}^{\top}\right]^{-1}S \right)$$

To get conditional distribution of parameters corresponding to stage 2 of the model, we assume that, $G_g Z_t = \beta_1 \tilde{Z}_{1t} + \beta_2 \tilde{Z}_{2t} + \beta_3 \tilde{Z}_{3t}$, with $\tilde{Z}$ being the proper linear combination of latent state $Z_t$ that affects the common value $W_{gt}$. Let us also denote the vector of parameters, 

$$\theta = \begin{bmatrix}
    \Phi_0 & \Phi & \beta_1 & \beta_2 & \beta_3 & H
\end{bmatrix}$$

Based on this, we have;

$$\theta \mid \text{rest} \sim \mathcal{MVN}\left( A^{-1}b; A^{-1}R \right)\text{ where }b = \left(
    \sum_{g, t} W_{gt},
    \sum_{g, t} W_{gt} W_{g(t-1)},
    \sum_{g, t} W_{gt} \tilde{Z}_{1t},
    \dots \right)^\top $$

and 

$$A = \begin{bmatrix}
    \sum_t N_t & \sum_{g, t} W_{g(t-1)} & \sum_t N_t\tilde{Z}_{1t} & \sum_t N_t\tilde{Z}_{2t} & \sum_t N_t\tilde{Z}_{3t} & \sum_{g, t} X_{gt}\\
    \sum_{g, t} W_{g(t-1)} & \sum_{g, t} W^2_{g(t-1)} & \sum_{g, t} W_{g(t-1)}\tilde{Z}_{1t} & \dots & \dots & \dots\\
    \sum_t N_t\tilde{Z}_{1t} & \sum_{g, t} W_{g(t-1)}\tilde{Z}_{1t} & \sum_t N_t\tilde{Z}^2_{1t} & \dots & \dots & \dots\\
    \dots & \dots & \dots & \dots & \dots & \dots \\
\end{bmatrix}$$

where $\mathcal{MVN}$ denotes the multivariate normal distribution.

\subsubsection{Estimated values and implications}
The performance of the estimated model has been shown in Figure \ref{fig:TSLH-val-train-test} and in Figure \ref{fig:TSLH-price-train-test}. Some of the residual diagnostics are shown in Figure \ref{fig:TSLH-val-res} and Figure \ref{fig:TSLH-price-res}, from which it is obvious that the residuals in logarithm scale follow an approximate normal distribution, other than some outlying values in the tail, as well as the residuals in the original scale of price, shows a histogram of leptokurtic distribution which closely resembles a lognormal one. The estimated values of the parameters of TSLH model are shown in the table \ref{tbl:3-stage-parameters}.

\begin{longtable}[h]{|c| L{5.5cm} | C{2cm} | C{2cm} | C{2cm} |}
    \caption{Description and Estimated values of the Parameters of Three Stage Latent Hierarchical (TSLH) Model (For sake of simplicity, only parameters for second and third stages are shown)}\label{tbl:3-stage-parameters}\\
    \hline
    \textbf{Parameter} & \centering {\textbf{Description}} & \textbf{95\% Lower Confidence Limit} & \textbf{Estimate} & \textbf{95\% Upper Confidence Limit}\\
    \hline
    \endfirsthead
    \multicolumn{5}{c}%
    {\tablename\ \thetable\ -- \textit{Continued from previous page}} \\
    \hline
    \textbf{Parameter} & \centering{\textbf{Description}} & \textbf{95\% Lower Confidence Limit} & \textbf{Estimate} & \textbf{95\% Upper Confidence Limit}
    \\
    \hline
    \endhead
    \hline \multicolumn{5}{r}{\textit{Continued on next page}}\\
    \endfoot
    \hline
    \endlastfoot
    $\Phi_0$ & Intercept in the second stage of the TSLH model. It incorporates bias in the estimation of latent common utility value $W_{gt}$ based on its previous value and exogenous variables & $0.1507994$ & $0.2425758$ & $0.3131058$\\
    \hline
    $\Phi$ & The effect of value of the latent common utility of previous time point has on explaining the value for current time point & $0.9623784$ & $0.9849693$ & $1.047248$\\
    \hline
    $\beta_1$ & The direct effect on the value of latent common utility $W_{gt}$ by the latent market demand $Z_t$ &  $-0.9382417$ & $0.00435303$ & $0.9869709$\\
    \hline
    $\beta_2$ & The substitution effect on the value of latent common utility $W_{gt}$ by the latent market demand $Z_t$ having the same grade cluster as tea grade in the index $g$ & $-0.4668122$ & $-0.0008195$ & $0.4655594$\\
    \hline
    $\beta_3$ & The substitution effect on the value of latent common utility $W_{gt}$ by the latent market demand $Z_t$ having the same source cluster as tea garden in the index $g$ & $-0.3721541$ & $-0.0028576$ & $0.4130075$\\
    \hline
    $H_1$ & Effect of logarithm of the total volume of the tea having the grade and garden characteristics same as $g$ coming in previous week, on explaining $W_{gt}$ & $-0.2156821$ & $-0.1529215$ & $-0.0791317$\\
    \hline
    $H_2$ & Effect of logarithm of the total volume of the tea having the tea grade cluster characteristics same as $g$, coming in previous week, on explaining $W_{gt}$ & $-0.3688453$ & $-0.321775$ & $-0.2688594$\\
    \hline
    $H_3$ & Effect of logarithm of the total volume of the tea having the tea source cluster characteristics same as $g$, coming in previous week, on explaining $W_{gt}$ & $-0.1018715$ & $-0.0534532$ & $-0.0070205$\\
    \hline
    $H_{4C}$ & Effect of clonal type of tea garden on $W_{gt}$ compared to the effect of regular type tea garden & $0.02484161$ & $0.0470203$ & $0.06408$ \\
    \hline
    $H_{4G}$ & Effect of gold type of tea garden on $W_{gt}$ compared to the effect of regular type tea garden & $-0.0386021$  & $-0.0252078$ & $-0.0203631$\\
    \hline
    $H_{4R}$ & Effect of royal type of tea garden on $W_{gt}$ compared to the effect of regular type tea garden & $-0.1240205$ & $0.007148$  & $0.244879$\\
    \hline
    $H_{4S}$ & Effect of special type of tea garden on $W_{gt}$ compared to the effect of regular type tea garden & $0.0316522$ & $0.034810$ & $0.037769$\\
    \hline
    $R$ & The variance of error in second stage of TSLH model & $0.144291$ & $0.146503$ & $0.1483624$\\
    \hline
    $\Gamma_{v}$ & Effect of logarithm of the net volume of the tea packet to be sold on the logarithm of Valuation of that particular packet & $0.7726263$ & $0.7829123$ & $0.7919209$\\
    \hline
    $S_{v}$ & The variance of error in third stage of TSLH model for estimating logarithm of Valuation of tea packets & $0.08097662$ & $0.08535738$ & $0.08893837$\\
    \hline
    $\Gamma_{pv}$ & Effect of logarithm of Valuation of a tea packet on the logarithm of Price of that particular tea packet & $0.966256$ & $1.006905$ & $1.040504$ \\
    \hline
    $\Gamma_{p}$ & Effect of logarithm of the net volume of the tea packet to be sold on the logarithm of Price of that particular packet & $-0.0370133$ & $-0.0190927$ &$-0.001596$\\
    \hline
    $S_{p}$ & The variance of error in third stage of TSLH model for estimating logarithm of Valuation of tea packets & $0.02646$ & $0.048802$ & $0.090898$\\
    \hline
\end{longtable}

\begin{figure}[ht]
\centering
\begin{subfigure}{.45\textwidth}
  \centering
  \includegraphics[width=\textwidth]{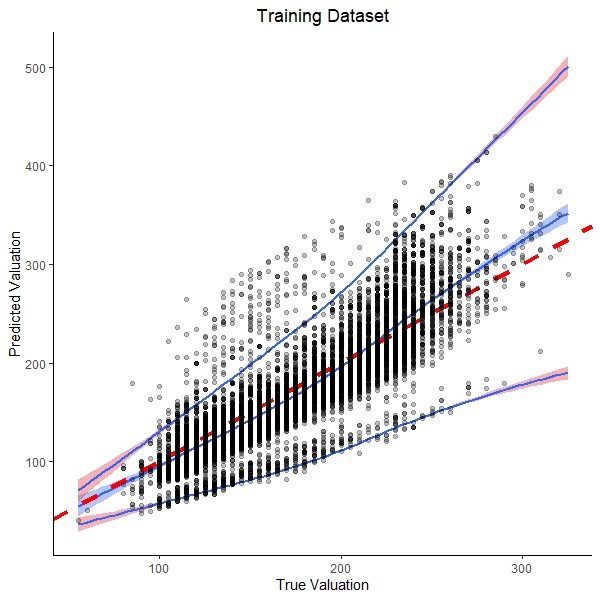}
\end{subfigure}%
\hfill
\begin{subfigure}{.45\textwidth}
  \centering
  \includegraphics[width=\textwidth]{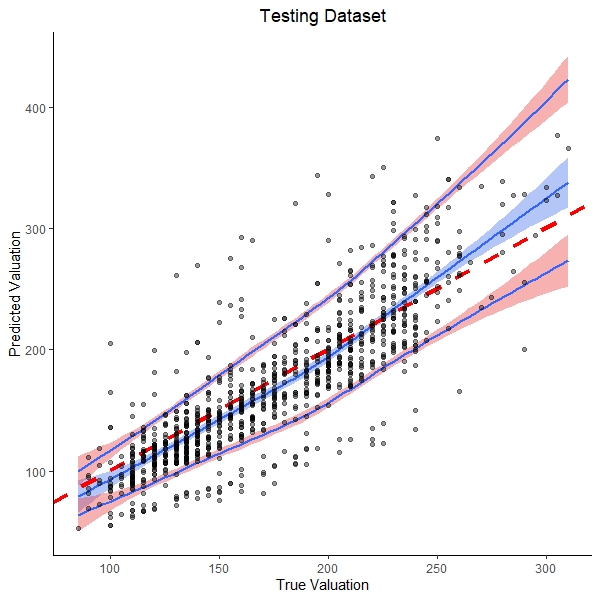}
\end{subfigure}
\caption{True Valuation vs Predicted Valuation in the training set and in the testing (or cross-validation) set. A kernel smoother for lower and upper limits of the $95\%$ confidence interval has been shown, along with the red dashed reference line $y = x$.}
\label{fig:TSLH-val-train-test}
\end{figure}

\begin{figure}[ht]
\centering
\begin{subfigure}{.45\textwidth}
  \centering
  \includegraphics[width=\textwidth]{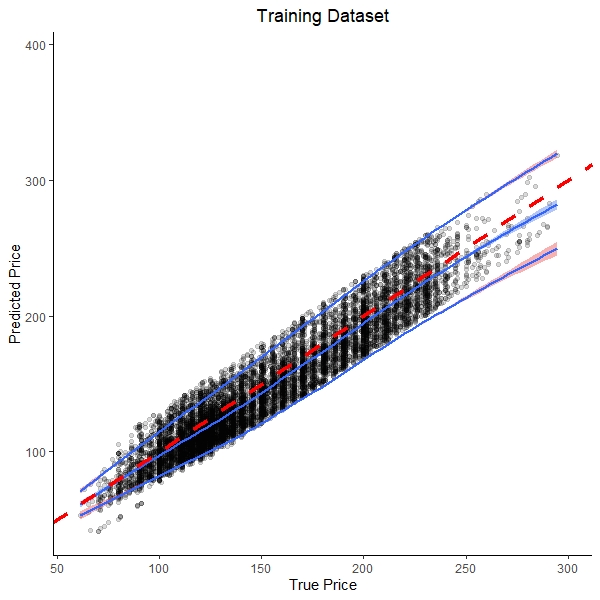}
\end{subfigure}%
\hfill
\begin{subfigure}{.45\textwidth}
  \centering
  \includegraphics[width=\textwidth]{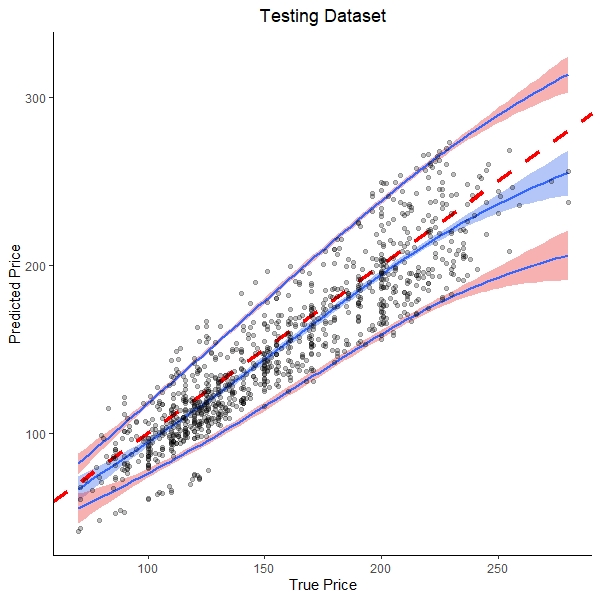}
\end{subfigure}
\caption{True Price vs Predicted Price in the training set and in the testing (or cross-validation) set. A kernel smoother for lower and upper limits of the $95\%$ confidence interval has been shown, along with the red dashed reference line $y = x$.}
\label{fig:TSLH-price-train-test}
\end{figure}

\begin{figure}[ht]
\centering
\begin{subfigure}{.45\textwidth}
  \centering
  \includegraphics[width=\textwidth]{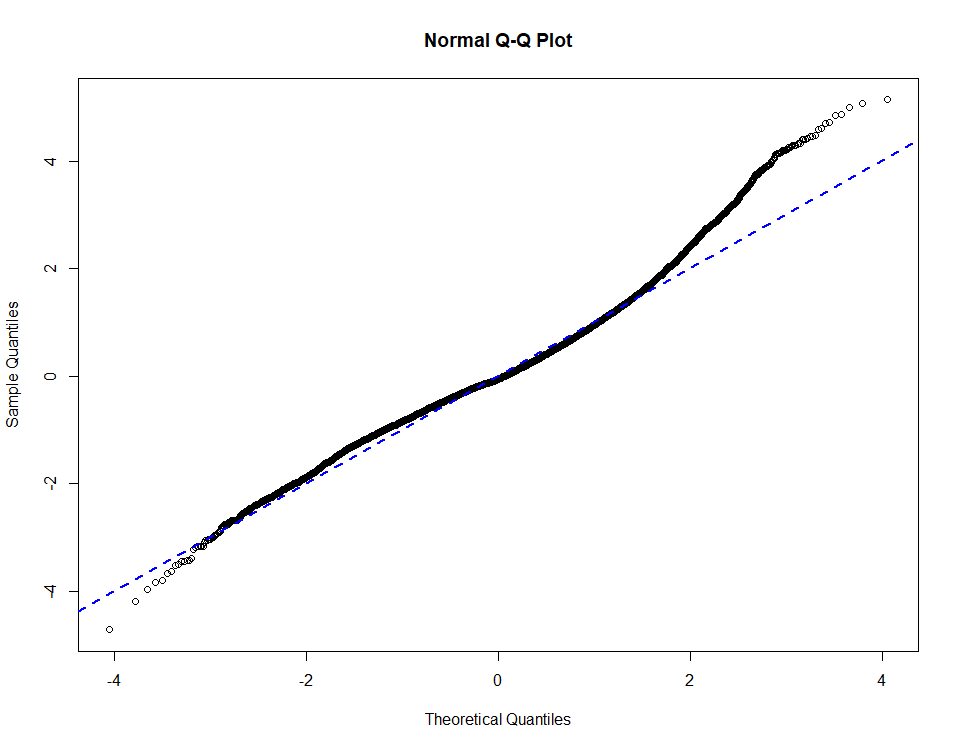}
  \caption{Normal Q-Q plot for residuals in estimating valuation}
\end{subfigure}%
\hfill
\begin{subfigure}{.45\textwidth}
  \centering
  \includegraphics[width=\textwidth]{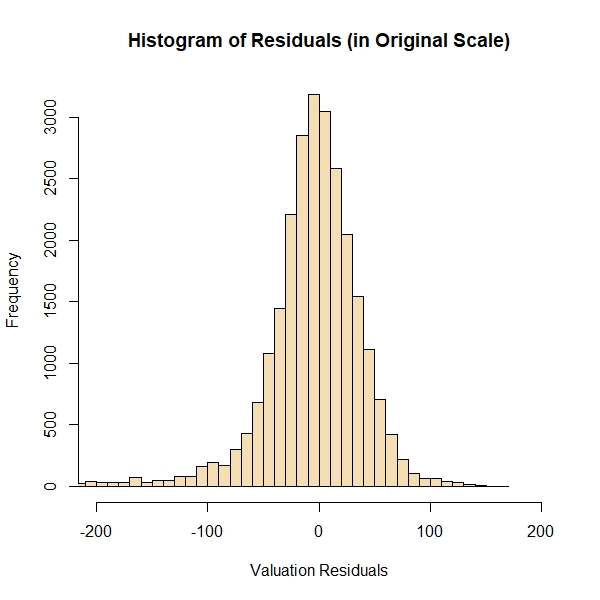}
  \caption{Histogram of the residuals of valuation model in estimating auctioneer's valuation}
\end{subfigure}
\caption{Residual diagnostics plots for valuation predicting component of TSLH model}
\label{fig:TSLH-val-res}
\end{figure}

\begin{figure}[ht]
\centering
\begin{subfigure}{.4\textwidth}
  \centering
  \includegraphics[width=\textwidth]{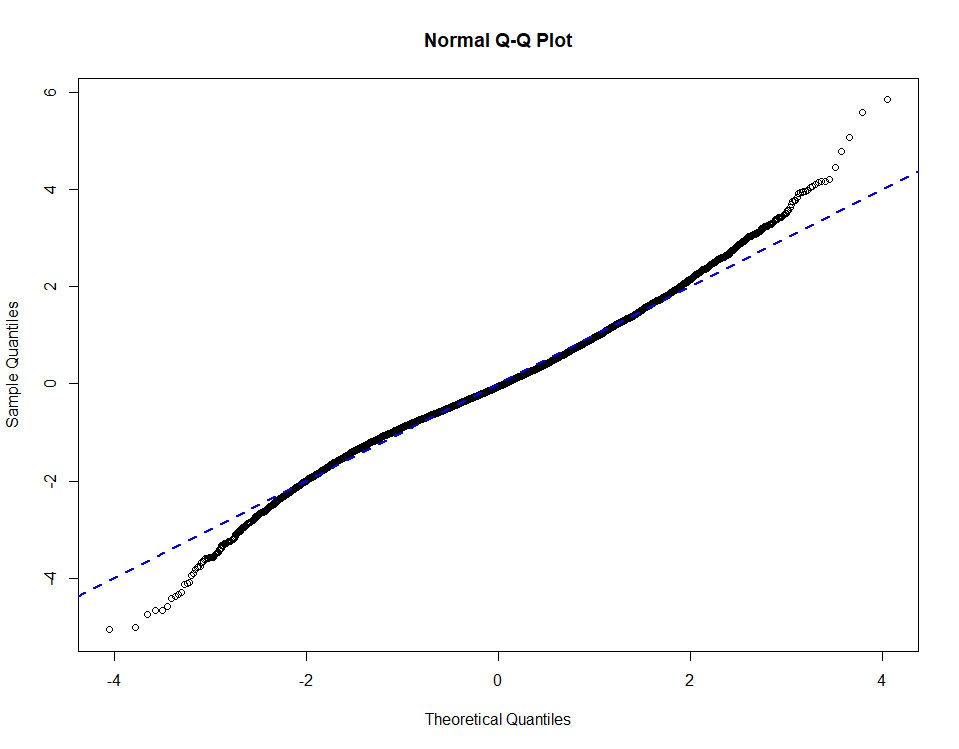}
  \caption{Normal Q-Q plot for residuals in estimating price}
\end{subfigure}%
\hfill
\begin{subfigure}{.5\textwidth}
  \centering
  \includegraphics[width=\textwidth]{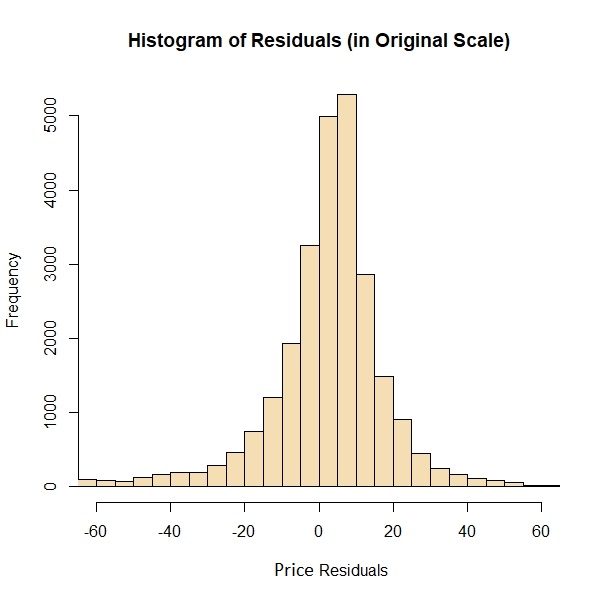}
  \caption{Histogram of the residuals of valuation model in estimating price}
\end{subfigure}
\caption{Residual diagnostics plots for price predicting component of TSLH model}
\label{fig:TSLH-price-res}
\end{figure}

To emphasize the goodness of fit for the Three Stage Latent Hierarchical (TSLH) model we obtain the following:

\begin{enumerate}
    \item The residuals in valuation predicting component of the TSLH model makes approximately $38$ rupees of error with $80\%$ confidence and about $59$ rupees of error with $95\%$ of confidence.
    \item The residuals in price predicting component of the TSLH model makes approximately $22$ rupees of error with $80\%$ confidence and about $29$ rupees of error with $95\%$ of confidence.
\end{enumerate}

The results we obtain have the following notable implications:

\begin{enumerate}
    \item Since $\beta_1$ is positive, it appears that the more is the market demand for a particular type of tea, the more is its value to the buyers. However, its small value (inclusion of 0 in confidence interval) suggests that this dependence with the underlying market condition is nonetheless small.
    \item Since $\beta_2$ and $\beta_3$ are estimated to be negative, it appears that the substitution effect is present in this scenario. The more is the market demand for the substitutes of a tea type, the less is its utility to the buyers, who are going to ultimately sell it to the consumers.
    \item Negative values of $H_1, H_2$, and $H_3$ show that, if one type of tea has already dispatched in the market by a large volume in the recent past, then their utility to the buyers of the auction is going to be less.
    \item An interesting thing to note is the opposite behavior of $\Gamma_p$ and $\Gamma_v$ when we allow an unrestricted $\Gamma_{pv}$ in the model. Generally, we expect that larger volumes of tea packets to be sold at smaller per unit price, which is in accordance with a negative $\Gamma_p$ value. However, if the valuation is indeed nothing but an educated guess for the price, then the behavior for $\Gamma_p$ and $\Gamma_v$ should be identical. The fact that $\Gamma_v$ is positive suggests that the auctioneer tries to balance for the act of the decrease in per-unit price with larger packet volumes, by setting its valuation at a price larger than the market would expect it to be.
    \item A large positive value of $\Gamma_{pv}$ suggests a very strong positive dependence on the prices of tea packets on the auctioneers' valuation of the tea packets. However, this dependence contains both the direct dependence of price on the valuation and an indirect dependence of price through the mediating effect of the spurious latent variable $W_{gt}$.
\end{enumerate}

\subsubsection{Test of causality}
Now, to answer the question about whether there is a significant direct effect from valuation to the price, (i.e. whether the dotted arrow in Figure \ref{fig:causal-dag-3-stage} exists or not) a very general approach is the conditional independence test, as discussed in \cite{pearl2016causal} and \cite{AnIntroductiontoCausalInference}. As shown in Figure \ref{fig:causal-dag-3-stage}, $W_{gt}$ and $u_{igt}$ creates a fork with $\log(\text{Valuation}_{igt})$ and $\log(\text{Price}_{igt})$ nodes in the DAG. However we shall require the following theorem: 

\begin{theorem}

The estimated $\Gamma_{pv}$ based on equation \ref{eqn:model-yigt} where $W_{gt}$ is substituted by the samples obtained from the Gibbs sampler, converges in distribution to the $\Gamma_{pv}$ estimated using the true value of the latent variables (although unknown).
\end{theorem}

\begin{proof}

We shall use $f(\cdot)$ and $F(\cdot)$ as a generic symbol for density and distribution function of a probability distribution, respectively.

Let, $W_{gt, k}^{(r)}$ be the samples taken from the simulated $k$-th Markov chain in Gibbs sampling, after $r$ iterations (i.e. after $r$ transition steps of the Markov chain), for the latent variable $W_{gt}$. It is well known that under certain regularity conditions such as positivity and connected of the conditional densities the simulated Markov chain will converge to a stationary distribution which is the desired posterior distribution, as shown in \cite{robert2013monte}. Therefore, as $r\rightarrow \infty$, we should have $W_{gt, k}^{(r)} \xrightarrow{\mathcal{L}} W_{gt, k}$, where $W_{gt, k}$ is a random variable following the posterior distribution with density $f(W_{gt} \mid \text{data})$. Here, the notation $\xrightarrow{\mathcal{L}}$ is used to denote convergence in distribution for random variables.

Hence, the joint distribution of the observed data together with the posterior samples from Gibbs sampling, converges to the joint distribution of the observed data together with unobserved latent variables.

\begin{align*}
& F\left(y_{igt}, u_{igt}, W_{gt, k}^{(r)}; \forall i, g, t\right) \\
= \quad & F\left(y_{igt}, u_{igt}; \forall i, g, t\right)
F\left(W_{gt}^{(r)} \mid y_{igt}, u_{igt} ; \forall i, g, t\right)\\
\xrightarrow{\mathcal{L}} \quad & F\left(y_{igt}, u_{igt}; \forall i, g, t\right)
F\left(W_{gt} \mid y_{igt}, u_{igt} ; \forall i, g, t\right)\\
= \quad & F\left(y_{igt}, u_{igt}, W_{gt, k}; \forall i, g, t\right)
\end{align*}

where we use the fact that $W_{gt}^{(r)} \xrightarrow{\mathcal{L}} W_{gt}$ means the distribution function converges pointwise.

Now the proof follows once we note that, $\Gamma_{pv}$, as a regression estimate, is a continuous function of $y_{igt}, u_{igt}$ and $W_{gt}$. Therefore, by continuous mapping theorem,

\begin{equation}
    \left[\sum_{i, g, t} u_{igt} u_{igt}^{\top}\right]^{-1} 
\left[\sum_{i, g, t} \left( y_{igt} - W_{gt, k}^{(r)}\mathbf{1}_2 \right)u_{igt}^{\top}\right] 
\xrightarrow{\mathcal{L}} 
\left[\sum_{i, g, t} u_{igt} u_{igt}^{\top}\right]^{-1} 
\left[\sum_{i, g, t} \left( y_{igt} - W_{gt}\mathbf{1}_2 \right)u_{igt}^{\top}\right]
\label{eqn:mean-converge}
\end{equation}

Let, $\hat{\Gamma}$ be the regression estimate if the true value of $W_{gt}$ were known, while $\tilde{\Gamma}_{k}^{(r)}$ be the regression estimate where $W_{gt}$ is substituted by the posterior sample $W_{gt, k}^{(r)}$.

Then, equation \ref{eqn:mean-converge} simply means $\tilde{\Gamma}_{k}^{(r)} \xrightarrow{\mathcal{L}} \hat{\Gamma}$, and hence particularly one entry of the matrix $\tilde{\Gamma}_{pv, k}^{(r)} \xrightarrow{\mathcal{L}} \hat{\Gamma}_{pv}$.
Thus, the result follows.
\end{proof}

There are particularly two remarks to be made relating to the above theorem.

\begin{enumerate}
    \item $\tilde{\Gamma}_{pv, k}^{(r)}$ is not the posterior samples obtained from the Gibbs sampling, but rather a regression estimate based on the posterior sample $W_{gt, k}^{(r)}$.
    \item Using this theorem, we know that under the null hypothesis that valuation has no direct causal effect on price in the TSLH model, a test for $H_0: \Gamma_{pv} = 0$ can be conducted using the regression estimate based on the proxies of the latent variables obtained from Gibb's sampler values. In other words, for each Markov chain, we can obtain samples $\tilde{\Gamma}_{pv, k}^{(r)}$, and as $r \rightarrow \infty$, these samples behave like independent and identically distributed samples from the distribution of $\hat{\Gamma}_{pv}$. Therefore, we can simply use the Hybrid confidence sets obtained from the Gibbs sampling to test our hypothesis.
\end{enumerate}

To test the causality from $\log(\text{Valuation}_{igt})$ to $\log(\text{Price}_{igt})$, we shall require a conditional independence test between these two variables conditioned on the value of $W_{gt}$ and $u_{igt}$. In the given model, such independence would hold if and only if $\Gamma_{pv} = 0$. Therefore, in view of the above theorem, using the posterior samples $W_{gt, k}^{(r)}$, we obtain the estimate of $\hat{\Gamma}$ as $1.006317$, and the $95\%$ confidence set turns out to be $(0.981592, 1.020155)$, which does not contain $0$, thereby, showing sufficient evidence against the null hypothesis of conditional independence.

On the other hand, based on the fitted model, let us consider the residuals from Valuation predicting component, and the residuals from Price predicting component (leaving Valuation as an explanatory variable), and denote their product moment correlation as $r_{pv}$. In other words,

$$r_{pv} = \text{cor}\left( \log(\text{Valuation}_{igt}) - W_{gt} - \Gamma_v \log(\text{Volume}_{igt}), \log(\text{Price}_{igt}) - W_{gt} - \Gamma_p \log(\text{Volume}_{igt}) \right)$$

Tracking this correlation where the latent variables and parameters are substituted by posterior samples obtained from the Gibbs sampler, we obtain the posterior mean of $r_{pv}$ as $0.7819911$. In contrast to that, the Pearson's correlation coefficient between the logarithms of those variables valuation and price is $0.9573915$, and the Pearson's correlation coefficient between those variables valuation and price without any transformation is $0.9609977$. Note that, the correlation between the residuals obtained from the causal linear model described before was $0.859$. Therefore, we find that, the linear model was enough to structurally model some of the dependence, while the TSLH model was further able to reduce the correlation by explaining temporal dependence structure within the data. However, there was still unexplained correlation, which was simply a manifestation of the causal relationship between auctioneers' valuation and the ultimate selling price at the auction.

\section{Remarks on Automation of the Process and Conclusion}
The preceding sections show us the utmost significance of the manual valuation of the tea packets that come in, in predicting the final price level. Thus the hope of automating the entire procedure seems unrealistic.

However, since this valuation is based on the inherent characteristics of the tea dust packets, and the volume of packets that arrive, we strongly believe that the practice of using valuations to set base prices can be done away with. George and Hui, in their paper \cite{optimal}, provide an ingenious way to estimate demand in the auction market, under the Independent Private Value Model second-price auctions. A generalization of this method, in this regard, to the Common Value (CV) \cite{auction_book} case, where the optimal symmetric bidding strategies are not the bidder signals themselves, but a monotonic function of their signals, may be helpful in our case. Then, knowing the distribution of the bidder values, an optimal reserve price may be set to maximize the \textit{ex-ante} expected revenue, which is a function of this distribution. \par
Levin and Smith, in their paper \cite{disproof}, have shown that under the non-IPV case, the optimal reserve price for the seller converges to her true value - here it's her manufacturing costs. Hence if the pool of bidders grow, then it would be safe for the seller to set the reserve price at her manufacturing costs.\par
Collusion among the bidders is often a very practical problem to ponder about, and most methodologies fail under scenarios not robust to such behavior. 
For example, in second-price auctions, one source of asymmetric equilibrium, (when the distribution has a support $[0,\omega]$) is for one bidder to bid $\omega$ and the others to bid $0$(or the minimum possible price). This is often realized in real-life scenarios, e.g. spectrum auctions. This is a possible solution here, and given that several auctions occur regularly in this market, bidders can sequentially alternate the role of the highest bidder, and thus can all be better off, at the cost of the seller. Our pricing model provides a way to detect such behavior on the part of the bidders. Since our pricing model with valuations provides an excellent fit to the true prices, this can be used to detect collusion. As in the case of collusion, the final price of the transaction would be low compared to the expected transaction price, a large deviation from the predictions would indicate the presence of such collusion. Our prices, with the estimated parameters, approximately follow a normal distribution, thus a low p-value from this distribution could be used as an indication of collusion of bidders.\par 
Further research on the aforesaid aspects could bring exciting breakthroughs in the path of automation.

{\bibliographystyle{unsrt}
\bibliography{maindoc.bbl}}

\end{document}